\documentclass[11pt, a4paper, oneside,reqno]{amsart}

\makeatletter
\def\@settitle{\begin{center}%
		\baselineskip14\p@\relax
		\normalfont\LARGE\bfseries
		\@title
	\end{center}%
}

\def\section{\@startsection{section}{1}%
	\z@{.7\linespacing\@plus\linespacing}{.5\linespacing}%
	{\normalfont\large\bfseries}}

\def\subsection{\@startsection{subsection}{2}%
	\z@{.5\linespacing\@plus.7\linespacing}{.5\linespacing}%
	{\normalfont\bfseries}}

\def\@setauthors{%
  \begingroup
  \def\thanks{\protect\thanks@warning}%
  \trivlist
  \centering\footnotesize \@topsep30\p@\relax
  \advance\@topsep by -\baselineskip
  \item\relax
  \author@andify\authors
  \def\\{\protect\linebreak}%
  \authors%
  \ifx\@empty\contribs
  \else
    ,\penalty-3 \space \@setcontribs
    \@closetoccontribs
  \fi
  \endtrivlist
  \endgroup
}

\makeatother

\usepackage[square,numbers]{natbib}
\usepackage{xcolor}
\usepackage{graphicx}
\usepackage{hyperref}
\hypersetup{hidelinks}
\usepackage{multirow} 
\usepackage{bbm}
\usepackage{tikz-cd}
\usepackage[ruled,linesnumbered, noend]{algorithm2e}
\usepackage[font=small,labelfont=rm]{subcaption}
\usepackage[font=small,margin=10pt]{caption}
\usepackage{amsmath,amssymb,amsfonts}
\usepackage{bm,mathtools}
\usepackage{amsthm}

\theoremstyle{plain}
\newtheorem{theorem}{Theorem}[section] % reset every section
\newtheorem{proposition}[theorem]{Proposition}
\newtheorem{example}{Example}
\newtheorem{lemma}[theorem]{Lemma}

\newtheorem{definition}{Definition}

\newtheorem{remark}{Remark}
\newtheorem{problem}{Problem}

\def\probMeas{{\mathcal{P}}}
\def\discMeas{{\mathcal{D}}}

\def\expect{{\mathbb{E}}}

\def\prob{\mathbb{P}}
\def\probHat{\hat{\mathbb{P}}}

\def\signature{{\Delta}}
\def\probQ{\mathbb{Q}}
\def\probD{\mathbb{D}}

\def\wasserstein{{\mathbb{W}}}

\newcommand{\arginf}[1]{\underset{#1}{\mathrm{arg inf}\ }}

\def\indicator{{\mathbbm{1}}}

\newcommand{\norm}[1]{\left\lVert#1\right\rVert}

\def\lipschitz{{\mathcal{L}}}

\def\boundProbDef{{\mathcal{W}_{\bm{\mathcal{R}}, \bm{\mathcal{C}}}}}

\def\realNum{{\mathbb{R}}}

\def\natNum{\mathbb{N}}

\def\sSimplex{\Pi}

\def\sB{{\mathcal{B}}}
\def\sC{{\mathcal{C}}}
\def\sD{{\mathcal{D}}}

\def\sL{{\mathcal{L}}}

\def\sN{{\mathcal{N}}}

\def\sP{{\mathcal{P}}}

\def\sR{{\mathcal{R}}}
\def\sS{{\mathcal{S}}}

\def\sW{{\mathcal{W}}}
\def\sX{{\mathcal{X}}}
\def\sY{{\mathcal{Y}}}

\def\bsR{{\bm{\mathcal{R}}}}
\def\bsC{{\bm{\mathcal{C}}}}
\def\bsS{{\bm{\mathcal{S}}}}

\def\mA{{{A}}}

\def\vpi{{{\pi}}}

\def\vb{{{b}}}
\def\vc{{{c}}}

\def\vx{{{x}}}

\newcommand{\evpi}[1]{{\vpi^{(#1)}}}

\newcommand{\SA}[1]{\textcolor{orange}{[SA: #1]}}

\newif\ifdoublecolumn
\doublecolumnfalse
\newcommand{\maybenewline}{\ifthenelse{\boolean{doublecolumn}}{\\}{}}
\newcommand{\maybeamp}{\ifthenelse{\boolean{doublecolumn}}{&}{}}
\newcommand{\maybeqquad}{\ifthenelse{\boolean{doublecolumn}}{\qquad}{}}
\newcommand{\maybenonumber}{\ifthenelse{\boolean{doublecolumn}}{\nonumber}{}}

\title[Efficient Uncertainty Propagation with Guarantees in Wasserstein Distance]{Efficient Uncertainty Propagation \\ with Guarantees in Wasserstein Distance}

\author{Eduardo Figueiredo$^{*,1}$, Steven Adams$^{*,1}$, Peyman Mohajerin Esfahani$^{1,2}$, Luca Laurenti$^{1}$
}
\thanks{$^*$ Equal contribution. The authors are with (1) Delft University of Technology, and (2) the University of Toronto. Corresponding author's email: e.figueiredo@tudelft.nl. This research is partially supported by the NWO (grant OCENW.M.22.056) and the ERC Starting Grant TRUST-949796.}

\begin{document}

\maketitle

\begin{abstract}
    %Propagating uncertainty through non-linear functions is one of the key challenges to be addressed to make accurate predictions and synthesize robust control strategies for uncertain dynamical systems.
In this paper, we consider the problem of propagating an uncertain distribution by a possibly non-linear function and quantifying the resulting uncertainty. We measure the uncertainty using the Wasserstein distance, and for a given input set of distributions close in the Wasserstein distance, we compute a set of distributions centered at a discrete distribution that is guaranteed to contain the pushforward of any distribution in the input set. 
Our approach is based on approximating a nominal distribution from the input set to a discrete support distribution for which the exact computation of the pushforward distribution is tractable,  thus guaranteeing computational efficiency to our approach. Then, we rely on results from semi-discrete optimal transport and distributional robust optimization to show that for any $\epsilon > 0$ the error introduced by our approach can be made smaller than $\epsilon$. Critically, in the context of dynamical systems, we show how our results allow one to efficiently approximate the distribution of a stochastic dynamical system with a discrete support distribution for a possibly infinite horizon while bounding the resulting approximation error. 
We empirically investigate the effectiveness of our framework on various benchmarks, including a 10-D non-linear system, showing the effectiveness of our approach in quantifying uncertainty in linear and non-linear stochastic systems.  
%Uncertainty propagation is a fundamental challenge for the control of modern cyber-physical systems, which are not only affected by intrinsic randomness in system dynamics but also by the uncertainty introduced by the learning algorithms used to estimate their components or parameters. This uncertainty cannot be generally neglected and must be propagated through possibly non-linear functions - which remains an open research problem. In this work, we build an algorithmic framework to propagate uncertain distributions through non-linear functions with formal guarantees of correctness, by leveraging techniques from quantization theory, distributional robust optimization literature, and optimal transport. The key novelty of our approach is to provide an algorithm to push forward Wasserstein ambiguity sets by a (possibly) non-linear function resorting to its local properties. We also show how this framework can be applied to the analysis of stochastic dynamical systems. We empirically investigate the effectiveness of our framework by showing the resulting propagation of ambiguity sets using various functions taken from benchmarks from the control community.  
\end{abstract}

\section{Introduction}
Modern cyber-physical systems are commonly affected by various sources of \emph{uncertainty}. These include both the uncertainty caused by the intrinsic randomness in the system dynamics \cite{stengel1994optimal} and the uncertainty due to the use of statistical learning algorithms to estimate the unknown components/parameters of the system \cite{ljung2010perspectives, williams1995gaussian}. Consequently, it is common that mathematical models are not only stochastic, but the distribution of the various random variables are themselves uncertain \cite{tsiamis2019finite}. As a result, when these models are used in safety-critical applications, the resulting uncertainty cannot be neglected \cite{cheng2021limits} and must be propagated through possibly non-linear functions. For instance, this is the case for stochastic dynamical systems, where the input distribution and the distribution of the noise affecting the system are commonly estimated from data and need to be propagated through the system dynamics for multiple (possibly infinite) time steps \cite{arnold1995random}. %An underestimation of the resulting uncertainty may lead to catastrophic effects, e.g., when the resulting models are used in safety-critical applications \cite{cheng2021limits}. 
Unfortunately, how to propagate uncertain distributions through non-linear functions is still an open question. This leads to the main question in this paper: how can we efficiently propagate an uncertain distribution through a non-linear function with formal guarantees of correctness?

Propagating a distribution $\prob$ through a function $f$ is equivalent to computing the push-forward distribution of $\prob$ by $f$ denoted by $f\#\prob$, which in the context of stochastic dynamical systems is equivalent to computing the Chapman-Kolmogorov Equation \cite{pavliotis2014stochastic}. Unfortunately, in general, even when $\prob$ is known, computing $f\#\prob$ in closed form is not possible and requires approximations \cite{landgraf2023probabilistic}, such as moment matching \cite{deisenroth2011pilco} or discretization-based methods \cite{julier2004unscented}. Unfortunately, these techniques either come with no correctness guarantees or are too computationally demanding due to the need to discretize the full state space and do not support any uncertainty in $\prob$. When $\prob$ is uncertain, the problem is exacerbated by the additional challenge of dealing with a possibly infinite set of distributions that must all be propagated through $f$. %Relying on the connection with distributionally Robust optimization  \cite{ben2009robust, bertsimas2004price, mohajerin2018data},
While this problem is receiving increasing interest \cite{aolaritei2022distributional, ernst2022wasserstein, adams2024finite}, existing approaches are either limited to linear $f$ or lack formality and scalability.% or cannot be directly applied to contexts where sequential uncertainty propagation is required.

In this paper, given an uncertain distribution $\prob$ and a non-linear function $f$, we present a framework to efficiently approximate $f\#\prob$ via discrete distributions with formal quantification of the resulting uncertainty. To quantify the uncertainty, we rely on the Wasserstein distance \cite{villani2009optimal}. This choice is motivated by the properties of the Wasserstein distance (i.e., it is a metric, it bounds the distance of the moments of the distributions, and convergence in the Wasserstein distance guarantees weak convergence) and its connection with optimal transport, which allows us to devise particularly efficient algorithms to solve our problem. Our approach is based on the fact that the Wasserstein distance between a continuous and a discrete distribution can be characterized as the solution of a semi-discrete optimal problem for which optimal solutions can be efficiently computed \cite{peyre2019computational}.
By using this connection and using techniques from distributional robust optimization and stochastic optimization \cite{ben2009robust, bertsimas2004price, mohajerin2018data, gao2023distributionally}, we show that given a discrete distribution approximating $\prob$, the Wasserstein distance between the pushforward of $\prob$  by $f$  and of its discrete approximation can be efficiently bounded, even when $\prob$ is uncertain and $f$ non-linear. The resulting bound can then be minimized by appropriately selecting the support of the approximating discrete distribution. This allows us to derive an efficient algorithmic framework that, given an uncertain distribution $\prob$ and a non-linear function $f$ and a given error threshold $\epsilon>0$, returns a discrete distribution whose push-forward through $f$ is guaranteed to be closer than $\epsilon$ to $f\#\prob$.

We then show how our framework can be applied to formally approximate the state distribution of stochastic dynamical systems over time. We show that in contrast to existing results \cite{aolaritei2022distributional,figueiredo2024uncertainty}, our approach can be successfully applied to linear and non-linear systems 
and for both finite and infinite time prediction horizons. In particular, under relatively mild assumptions on $f,$ we prove the convergence of the approximation error of our approach over time to a fixed point. To further illustrate the usefulness of our framework, we perform an empirical evaluation on various benchmarks. In particular, we consider various linear and non-linear systems, including standard control benchmarks such as the Mountain Car \cite{singh1996reinforcement} and Dubins Car \cite{balkcom2018dubins}, and a 10-D model of a neural network. 
The empirical analysis highlights how our framework can successfully approximate the push-forward distributions in both linear and non-linear cases and with relatively small discrete distributions, thus showcasing its potential to efficiently approximate complex distributions even in complex iterative prediction settings.
    
%a novel algorithmic framework to obtain a tractable propagation of $\rho$-Wasserstein ambiguity sets through non-linear functions with guarantees of correctness, ii) a strategy to apply the framework to stochastic dynamical systems, with a discussion on its ergodic properties, and iii) empirical validation of our approach in several different contexts, including complex non-linear dynamical systems, such as the Mountain and Dubins Car.
In summary, the main contributions of this work are listed below:
\begin{itemize}
    \item {\bf Uncertainty propagation}: upper-bounds on the uncertainty measured in terms of the Wasserstein distance of the pushforward of an uncertain probability distribution through a possibly non-linear function (Theorem~\ref{th:bound}), and a refined version under no ambiguity (Theorem~\ref{th:bound-zero-ball});
    % We also show that this bound can be improved in the special case of propagating known probability distributions (Theorem~\ref{th:bound-zero-ball}).
    \item {\bf Algorithm \& convergence rate}: an efficient algorithmic procedure (Algorithm~\ref{alg:compute-bound}) to approximate the pushforward of an uncertain probability distribution by a discrete distribution, with guaranteed convergence in \(\rho\)-Wasserstein distance (Theorem~\ref{prop:convergence-algorithm-error-approx});
    % \item We specialize our framework to the the finite and infinite horizon propagation of uncertainty through non-linear dynamics.
    \item {\bf Approximation error dynamics}: an application of our framework to stochastic dynamical systems for both finite and infinite prediction horizons (Theorem~\ref{thm:how-to-propagate-stochastic-systems}).
\end{itemize}

The paper is organized as follows. We formulate the problem in Section~\ref{section:problem-statement}, present the formal uncertainty propagation error bounds in Section~\ref{section:wasserstein-bounds}, and introduce an algorithmic procedure to propagate the uncertain distributions and compute these bounds in Section~\ref{section:optimize-locs-convergence}.  
Finally, in Section \ref{section:experimental-results}, we conduct an extensive empirical validation on several benchmarks, including complex non-linear dynamical systems, such as the Mountain and Dubins Car, and a 10-D non-linear system.

\section{Related Works}
Our work is connected with the distributionally robust optimization literature. In distributionally robust optimization, one is usually interested in computing the worst expected value of a certain transformation of a random vector w.r.t. a family of distributions $\sP$, i.e., $\sup_{\prob \in \sP} \expect_{\xi\sim\prob}[f(\xi)]$ \cite{popescu2007robust, delage2010distributionally, calafiore2006distributionally, gao2023distributionally}. In particular, \cite{mohajerin2018data} provides techniques to characterize this worst-case expectation via convex optimization in the case where $\sP$ is defined as a Wasserstein ambiguity set. %The authors argue that Wasserstein ambiguity sets offer powerful out-of-sample performance guarantees \cite{mohajerin2018data, aolaritei2023capture}, which is one of the motivations behind our metric choice in this paper. 
Similarly, \cite{gao2023distributionally} prove that the distributionally robust problem is equivalent to a dual minimization program in $\realNum$-space for a large class of functions $f$ and spaces $\sP$, a result from which we took inspiration to demonstrate some of our results in this paper. In our work, however, we aim to find the worst \emph{Wasserstein distance} between push-forwarded measures, given that they belong to a given set of probabilities close in Wasserstein distance. Furthermore, the Wasserstein distance is defined not as an expectation on the $\sP$-space but as the infimum expectation of a specific cost function on the coupling space. Thus, a different framework must be devised to solve the problem. 

Uncertainty propagation for various classes of functions has also been studied in the context of dynamical systems \cite{landgraf2023probabilistic}.
In \cite{aolaritei2022distributional}, the authors provide a framework for the propagation of a set of distributions close in the Wasserstein distance in dynamical systems, where a distribution needs to be propagated through the system dynamics multiple times. These results have been applied in the context of stochastic model predictive control \cite{aolaritei2023wasserstein, mcallister2024distributionally}. However, in terms of numerical tractability, these techniques are specific to linear systems. 
Instead, in \cite{ernst2022wasserstein}, the authors consider the uncertainty propagation problem in the context of random differential equations. The resulting bounds, however, involve different Wasserstein spaces, i.e., they propose a bound of type $\wasserstein_{\rho_1}(f\#\prob, f\#\probQ) \leq C(f, \prob, \probQ) \wasserstein_{\rho_2}(\prob, \probQ)$ where\footnote{The term $C(f, \prob, \probQ)$ is a constant upper-bounding the moment under both $\prob$ and $\probQ$ of a function only requiring local Lipschitz continuity from $f$, which is a less restrictive assumption compared to the piecewise Lipschitz continuity that we need in our work.} $\rho_1 < \rho_2$, thus not allowing for its use in settings where the uncertainty must be propagated multiple times and the information on the moments must be conserved\footnote{The $\rho$-Wasserstein distance between $\prob$ and $\probQ$ is related to how close their $\rho$-moments are (\cite{villani2009optimal, adams2024finite}). Propagating a bound in the $\rho_2$-Wasserstein space to the $\rho_1$-Wasserstein space implies a loss of information on the difference of the higher moments of the push-forwarded measures.}. Uncertainty propagation in stochastic dynamical systems has also been considered in 
\cite{figueiredo2024uncertainty}, where the authors consider mixture approximations of the distribution of a dynamical system over time with bounds in total variation. However, the resulting bounds cannot be applied in the context of our paper where we approximate a continuous distribution with a discrete one,  grow linearly with time independently of $f$, and, consequently, become uninformative for a not small prediction time horizon.  %In contrast, we prove that for our approach, under mild contractility properties of the function where a distribution is propagated to, the error of iterative propagation of a distribution through the same function reaches a fixed point and consequently remains informative for possible infinite time horizon when applied for uncertain propagation in dynamical systems. 
A related work is also \cite{adams2024finite}, which views neural networks as stochastic dynamical systems and presents an algorithmic framework to approximate a stochastic neural network with a mixture of Gaussian distributions with error bounds in Wasserstein distance. This approach is, however, specific to neural networks.

Another related line of work is that of stochastic abstraction-based methods, where a stochastic system is abstracted into a variant of a discrete Markov chain \cite{abate2010approximate,kushner1990numerical} and that have also been recently extended to support distributional uncertainty on the system dynamics \cite{gracia2025efficient}.  However, these works suffer from the state-space explosion problem due to the need to finely discretize the full support of the distributions. In contrast, our approach approximates a continuous distribution with a discrete one by selecting the support of the discrete distribution to minimize the distance from the continuous one. This allows us to reduce the size of the support of the resulting discrete distribution by only placing locations in the regions with high probability mass.

\section{Preliminaries}
Here, we provide the necessary preliminaries on the Wasserstein distance and the quantization of probability distributions.

\subsection{Notation}
For a vector $x \in \realNum^d$,  we denote by $x^{(i)}$ its $i$-element. For a set $\sX \subseteq \realNum^d$, the indicator function for $\sX$  is denoted as $\indicator_{\sX}(x) \coloneqq 1 \text{ if } x \in \sX \text{; otherwise } 0$. For $\sX \subseteq \realNum^d$, we denote a partition of $\sX$ in $N$ \emph{regions} $\bsR \coloneqq \big\{ \sR_i \big\}_{i=1}^N$, i.e. $\sR_i \subseteq \sX$, $\bigcup_{i=1}^N \sR_i = \sX$, and  $\forall i \neq j, \sR_i \cap \sR_j = \emptyset$.
Given a Borel measurable space $\sX \subseteq \realNum^d$, we denote by $\mathcal{B}(\sX)$ the Borel sigma algebra over $\sX$ and by $\sP(\sX)$ the set of probability distributions on $\sX$.
For a random variable $x_{t}$ taking values in $\sX$, $\prob_{x_t} \in \sP(\sX)$ represents the probability measure associated to $x_t$. 
For $N \in \natNum$, $\sSimplex^N := \{ \pi \in \realNum^{N}_{\geq 0} \; : \; \sum_{i=1}^{N} \evpi{i} = 1 \}$ is the $N$-simplex. A discrete probability distribution $\probD \in \probMeas(\sX)$ is defined as $\probD=\sum_{i=1}^N\evpi{i}\delta_{\vc_i}$, where $\delta_\vc$ is the Dirac delta function centered at location $\vc \in \sX$ and  $\pi \in \sSimplex^N$ and $N$ is the number of locations in the support of $\probD$. The set of discrete probability distributions on $\sX$ with at most $N$ locations is denoted as $\discMeas_N(\sX)\subset \mathcal{P}(\sX)$. For a probability distribution ${\prob} \in \sP(\sX)$ and a measurable function $g: \sX \rightarrow \sY \subseteq \realNum^{q}$, we denote the push-forward measure of ${\prob}$ by $g$ as $g \# {\prob}$ such that for all $A \subset\mathcal{B}(\sY)$, $(g \# \prob)(A) \coloneqq \prob(g^{-1}(A))$. We note that $g \# \prob$ is still a probability distribution such that $g \# \prob \in \sP(\sY)$.

\subsection{Wasserstein (or Kantorovich) distance}
Let  $\rho \geq 1$, $\mathcal{X}\subseteq \mathbb{R}^d$, and define $\sP_{\rho}(\sX)$ as the set of probability distributions with finite $\rho$-th moments under the $L_\rho$-norm, i.e. all $\prob \in \sP(\sX)$ such that $\int_{\sX} \norm{x}^\rho d\prob(x) < \infty$.

\begin{definition}[$\rho$-Wasserstein distance]\label{def:wasserstein-distance}
For $\mathbb{P}, \mathbb{Q} \in \sP_{\rho}(\sX)$ the $\rho$-Wasserstein distance $\wasserstein_{\rho}$ between $\mathbb{P}$ and $\mathbb{Q}$ is defined as
\begin{equation} 
    \wasserstein_{\rho}(\prob, \probQ) \coloneqq \left( \inf_{\gamma \in \Gamma(\mathbb{P}, \mathbb{Q})} \int_{\sX \times \sX} \norm{x-y}^{\rho} d\gamma(x, y) \right)^{\frac{1}{\rho}}
\end{equation}
where $\Gamma(\mathbb{P}, \mathbb{Q}) \subset \sP_{\rho}(\sX \times \sX)$ represents the set of joint probability distributions with given marginals $\mathbb{P}$ and $\mathbb{Q}$ (also known as \emph{couplings} between $\mathbb{P}$ and $\mathbb{Q}$), i.e., for all $\gamma \in \Gamma(\mathbb{P}, \mathbb{Q})$, it holds that:
$$
\gamma(A \times \sX) = \mathbb{P}(A), \; \gamma(\sX \times A) = \mathbb{Q}(A) \qquad \forall A \in \sB(\sX).
$$
% and $\norm{.}$ is the $L_\rho$-norm with $1 \leq \rho < \infty$.
\end{definition}
Additionally, we define the $\rho$-Wasserstein ball of radius $\theta \geq 0$ centered at the probability distribution $\prob \in \sP_\rho(\sX)$, also called ambiguity set, as:
\begin{align}
    \mathbb{B}_\theta(\prob) \coloneqq \Big\{ \probQ \in \sP_\rho(\sX) : \wasserstein_\rho(\prob, \probQ) \leq \theta \Big\}.
\end{align}
That is, $ \mathbb{B}_\theta(\prob)$ contains all probability measures closer than $\theta$ to $\prob$ according to $\wasserstein_\rho$.

We finally state an identity that follows directly from the definition of the Wasserstein distance and will be extensively employed in the following sections:
\begin{align}\label{eq:wass-for-pushforward-equivalence}
    \maybeamp \wasserstein_\rho(f\#\prob, f\#\probQ) \maybenewline 
    \maybeamp\maybeqquad= \left( \inf_{\gamma \in \Gamma(\prob,\probQ)} \int_{\sX \times \sX} \norm{f(x)-f(y)}^{\rho} d\gamma(x, y) \right)^{\frac{1}{\rho}}. \maybenonumber
\end{align}
This identity states that the $\rho$-Wasserstein distance between the pushforward of two probability measures under a function $f$ is equal to a $\rho$-Wasserstein-like distance between the original measures, using a modified cost structure $\big(\norm{.} \circ (f \times f)\big)$.

\iffalse
    Finally, in Proposition \ref{prop:wass-for-pushforward-equivalence}, we present a result that will be extensively employed in the following sections. It states that computing the $\rho$-Wasserstein distance between the pushforward of two probability measures by a function $f$ is equivalent to computing a $\rho$-Wasserstein-like distance between the original measures in which we consider a different transportation cost structure ($\norm{.} \circ (f \times f)$).
    
    \begin{proposition}[Push-forward error\cite{aolaritei2022distributional}, Prop. 3]\label{prop:wass-for-pushforward-equivalence}
        Let $\prob, \probQ \in \sP_\rho(\sX)$. Consider a function $f: \sX \rightarrow \sY$. Then,
        \begin{equation}
            \wasserstein_\rho(f\#\prob, f\#\probQ) = \left( \inf_{\gamma \in \Gamma(\mathbb{P}, \mathbb{Q})} \int_{\sX \times \sX} \norm{f(x)-f(y)}^{\rho} d\gamma(x, y) \right)^{\frac{1}{\rho}}.
        \end{equation}
    \end{proposition}
\fi

\begin{figure}[ht]
    \centering
    \ifthenelse{\boolean{doublecolumn}}{
        \includegraphics[width=0.75\linewidth]{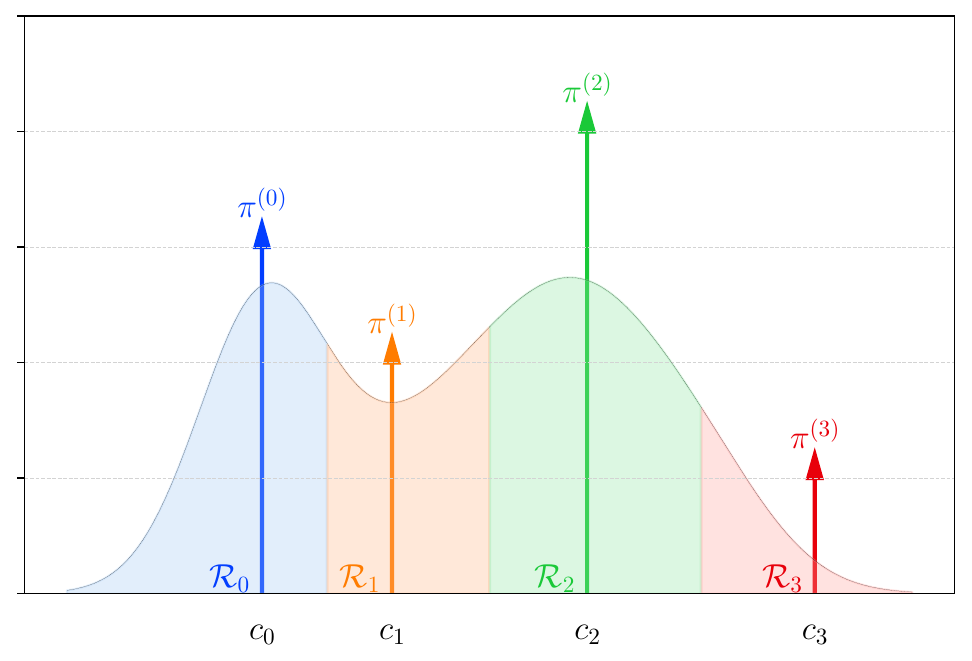}
        
    }{
        \includegraphics[width=0.5\linewidth]{Figures/Examples/signaturization.pdf}
    }
    \caption{Schematic representation of the density of a continuous probability distribution $\prob$, and its quantization $\Delta_{\bsR, \bsC}\#\prob$, which has support of size $N=|\bsC|=4$,  where we represent $\evpi{i} \coloneqq \prob(\sR_i)$, a notation that will be commonly adopted in the rest of the paper.}
    \label{fig:signaturization-representation}
\end{figure}
\subsection{Quantization of probability distributions}\label{subsection:algorithm-quantization-gaussian-case}
For $\sX \subseteq \realNum^d$, we consider a $\sX$-partition $\bsR = \big\{ \sR_i \big\}_{i=1}^N$ in $N$ \emph{regions}. Further, we denote by $\bsC = \big\{c_i\}_{i=1}^N$ a set of $N$ points in $\realNum^d$, which we refer to as \emph{locations} henceforward.

\begin{definition}[Quantization of a probability distribution]
    For partition $\bsR$ and set of locations $\bsC $, the quantization operator $\Delta_{\bsR, \bsC} : \sX \rightarrow \sX$ is defined by
    \begin{align}
        \Delta_{\bsR, \bsC}(x) \coloneqq \sum_{i=1}^N c_i \indicator_{\sR_i}(x).
    \end{align}
    That is, the quantization operator takes any point in the region $\sR_i$ and brings it to the location $c_i$. 
    For any probability distribution $\prob \in \sP(\sX)$, the \emph{quantization} (or \emph{discretization}) of $\prob$ is defined as the pushforward measure
    \begin{align}
        \Delta_{\bsR, \bsC} \# \prob = \sum_{i=1}^N \prob(\sR_i) \delta_{c_i} \in \discMeas_N(\sX).
    \end{align}
\end{definition}

Note that in the definition of $  \Delta_{\bsR, \bsC}$, we do not assume any relationship between the partition and the locations, although it is natural to pick $c_i \in \sR_i$. We should also stress that if, for a given set of locations $\bsC$, one defines the partition as the Voronoi partition w.r.t. $\bsC$, i.e., we take $\bsR$ with each region being constructed as
\begin{align}
    \sR_i \coloneqq \Big\{ z \in \realNum^d : \norm{z - c_i} \leq \norm{z - c_j}, \forall j \neq i \Big\},
\end{align}
where $\norm{.}$ is the underlying norm, then the quantization operator $\Delta_{\bsR, \bsC}$ is equivalent to the signature operation described in \cite{adams2024finite}. An example of the quantization operator is shown in Figure \ref{fig:signaturization-representation}. 

The concept of quantizing a continuous probability distribution is well known in the literature \cite{graf2000foundations, ambrogioni2018wasserstein, adams2024finite} and the following result to compute the $\rho$-Wasserstein distance between $\prob$ and $\Delta_{\bsR, \bsC} \# \prob$ is a straightforward extension of Theorem 1 of \cite{ambrogioni2018wasserstein}, which we will employ in this work.
%Further, in the case of a signature operation, the $\rho$-Wasserstein distance between $\prob$ and $\Delta_\bsC \# \prob$ can be solved (or upper bounded) explicitly \cite{ambrogioni2018wasserstein, adams2024finite}. In Proposition \ref{prop:compute-theta-d}, we extend Proposition 8 of \cite{adams2024finite} to generaldiscretizationss. 
\begin{proposition}[Quantization error]\label{prop:compute-theta-d}
    Let $\prob \in \sP_\rho (\sX)$ and assume a given $\sX$-partition $\bsR = \big\{ \sR_i \big\}_{i=1}^N$ and set of locations $\bsC = \big\{c_i\big\}_{i=1}^N$. Then, for any $\rho \geq 1$,
    \begin{equation}\label{eq:7}
        \wasserstein_\rho(\prob, \Delta_{\bsR, \bsC} \# \prob) \leq \bigg( \sum_{k=1}^N \int_{\sR_k} \norm{x-c_k}^\rho d\prob(x) \bigg)^\frac{1}{\rho}
    \end{equation}
    Furthermore, if $\bsR$ is chosen to be the Voronoi partition w.r.t. $\bsC$, then
    ~\eqref{eq:7} holds with equality. 
\end{proposition}
This result comes from the fact that the particular coupling which transports the probability mass of $\prob$ in the region $\sR_k$ to the location $c_k$ belongs to $\Gamma(\prob, \Delta_{\bsR, \bsC}\#\prob)$. Thus, its associated cost (right-hand side of ~\eqref{eq:7}) upper-bounds the Wasserstein distance between these two distributions.
\begin{remark}[Computing the quantization error]
% Constrained moments bound computation
\label{remark:theta-d-how-to-compute}
    To compute the constrained $\rho$-moments in \eqref{eq:7}, there are various approaches one can rely on. For instance, if $\prob$ is a product measure, $\norm{.}$ is the $L_{\rho}$-norm, and $\bsR$ is a set of axis-aligned hyper-rectangles, then we have:
    \begin{align*}
        \maybeamp \sum_{k=1}^N \int_{\sR_k} \norm{x-c_k}^\rho d\prob(x) \maybenewline \maybeamp = \sum_{k=1}^N \sum_{m=1}^d \int_{r_k^{(m)}} \left| x^{(m)} - c_k^{(m)} \right|^\rho d\prob_m\big( x^{(m)} \big)  \prod_{j \neq m} \prob_j \big(r_k^{(m)} \big),
    \end{align*}
    where $r_k^{(m)}\coloneqq [a_k^{(m)}, b_k^{(m)}]$, and $\sR_k = \prod_{m=1}^d r_k^{(m)}$. That is, we need to compute a set of constrained $\rho$-moment of the univariate distributions $\prob_m$, which is analytically tractable for many distributions (especially - although not limited to - for $\rho \in \{ 1, 2 \}$), including Gaussian (see Proposition 9 and Corollary 10 in \cite{adams2024finite} also for the general full covariance multivariate case), Uniform, Exponential, or Gamma distributions.
    Another particularly favorable case is when $\prob$ is discrete, i.e., $\prob \in \sD_N(\realNum^d)$; in this case, the bounds can be computed directly because of the finiteness of the support of the distributions.
\end{remark}

\section{Problem Formulation}\label{section:problem-statement}
After having formally defined $ \wasserstein_\rho$ and $\Delta_{\bsR, \bsC}$ we are now ready to state the main problem we consider in this paper. 
Given an uncertain distribution $\probQ\in \mathbb{B}_\theta(\prob)$, where $\mathbb{B}_\theta(\prob)$ is a Wasserstein ambiguity set of radius $\theta\geq 0$ centered at $\prob \in \sP_\rho(\mathcal{X})$ for $\mathcal{X}\subseteq \realNum^d$, and a possibly nonlinear measurable piecewise Lipschitz continuous function $f:\sX\to \sY$, our goal is to find discrete approximations of the pushforward distribution of
$\probQ$ by $f$. In particular, we consider the following problem.

\begin{problem}\label{prob:main}
    For an error threshold $\epsilon >0$, find a $\sX$-partition $\bsR = \big\{ \sR_i \big\}_{i=1}^N$ and locations $\bsC = \big\{ {c}_1, ..., {c}_N \big\}$ such that
    \begin{subequations}
    \begin{align}
    \label{eq:ProbConverg}
        \maybeamp\Bigg|  \sup_{\probQ \in \mathbb{B}_\theta(\prob)} \wasserstein_{\rho}(f \# \probQ, f \# \Delta_{\bsR, \bsC} \# \prob) \maybenewline\maybeamp - \sup_{\probQ \in \mathbb{B}_\theta(\prob)} \wasserstein_{\rho}(f \# \probQ, f  \# \prob) \Bigg| \leq \epsilon.  \maybenonumber 
        % \Bigg| \sup_{\probQ \in \mathbb{B}_\theta(\prob)} \wasserstein_{\rho}(f \# \probQ, f \# \Delta_{\bsR, \bsC} \# \prob) - \sup_{\probQ \in \mathbb{B}_\theta(\prob)} \wasserstein_{\rho}(f \# \probQ, f  \# \prob) \Bigg| \leq \epsilon. 
    \end{align}
    Furthermore, find a bound \(\boundProbDef \geq 0\) such that
    \begin{align}
    \label{Eqn:probBounds}
    \sup_{\probQ \in \mathbb{B}_\theta(\prob)} \wasserstein_{\rho}(f \# \probQ, f \# \Delta_{\bsR, \bsC} \# \prob) \leq \boundProbDef.
    \end{align}
    \end{subequations}
\end{problem}

%convergence will be by using triangle inequality by showing W(P,Q) is within W(Q,deltaP) +/- W(P,deltaP)
The goal of Problem \ref{prob:main} is to find arbitrarily accurate discrete approximations of the pushforward measure of an uncertain distribution and, critically, to bound the resulting uncertainty.  Note that the convergence requirement in \eqref{eq:ProbConverg} to $\sup_{\probQ \in \mathbb{B}_\theta(\prob)} \wasserstein_{\rho}(f \# \probQ, f  \# \prob)$ is natural as if $\prob$ and $\probQ$ differ, then, in general, the distance of their pushforward distributions will not be zero. Hence, even if $\epsilon$, the error introduced by the quantization, vanishes, then $\sup_{\probQ \in \mathbb{B}_\theta(\prob)} \wasserstein_{\rho}(f \# \probQ, f \# \Delta_{\bsR, \bsC} \# \prob)$ may not. The need to compute $\boundProbDef$ in Problem \ref{prob:main} guarantees that in this paper we are not only interested in computing converging discrete approximations, but also in obtaining non-trivial error bounds for the resulting uncertainty quantification problem.
%The choice of the $\wasserstein_\rho$ distance is motivated by the particularly favourable properties of the $\wasserstein_\rho$ distance, including that convergence in $\rho$-Wasserstein distance implies convergence of the first $\rho$-moments \cite{villani2009optimal}.  

Problem \ref{prob:main} aims at 
generalizing existing methods to perform uncertainty propagation of probability distributions, such as non-linear filtering \cite{bucy1971digital, alspach1972nonlinear} or sigma point methods \cite{julier2004unscented}, by computing formal error bounds on the error in terms of the Wasserstein distance and in selecting optimal discrete distribution approximations, which also accounts for the uncertainty in $\prob$. %In fact, we should stress once again that computing worst-case error bounds has become essential to apply approximate uncertainty propagation techniques in safety-critical tasks \cite{}. 
Note also that for $\theta=0,$ \eqref{Eqn:probBounds} reduces to bounding $\wasserstein_{\rho}(f \# \prob, f \# \Delta_{\bsR, \bsC} \# \prob)$, that is, the error in the pushforward approximation of a discrete operator applied to a known distribution. While this is itself an important open problem \cite{landgraf2023probabilistic}, as we will illustrate in Example \ref{example:initial}, we should already stress that in the case of uncertainty propagation in dynamical systems, which is the main application we consider in this paper, considering $\theta>0$ in \eqref{Eqn:probBounds} is essential.  

\begin{example}[Dynamical systems]
% Uncertainty propagation through dynamical systems]
\label{example:initial}
 Consider the general model of a discrete-time stochastic process
\begin{equation*}\label{eq:examples}
    x_{t+1} = f(x_t,w_t), \quad x_0 \sim \prob_{x_0}, w_t \sim \prob_{w_t},
\end{equation*}
where $\prob_{x_0}$ and $\prob_{w_t}$ represent, respectively, the distribution of the initial condition and of the noise affecting the system at time $t$. 
%System \eqref{eq:examples} represents a general model of possibly non-linear dynamical system.
If $f$ is non-linear or $\prob_{w_t}$ is not Gaussian, then the distribution of the system at time $t$, $\prob_{x_t}$, generally cannot be obtained in closed form and requires approximations \cite{girard2002gaussian, landgraf2023probabilistic,figueiredo2024uncertainty}. A solution to Problem \ref{prob:main} allows one to approximate arbitrarily well
$\prob_{x_t}$ for any $t>0$ with a discrete distribution $\hat{\prob}_{x_t}$ by iteratively approximating the pushforward distribution $f \# {\prob}_{x_t}$ and quantifying the approximation error. {Note that for $t>0$, the distribution of $\prob_{x_{t}}$ is uncertain because of the uncertainty introduced by the quantization at the previous time steps. Consequently, approximating $\prob_{x_{t+1}}$ requires one to consider $\theta>0$ in Problem \ref{prob:main} to propagate the resulting uncertainty through $f$}. In Section \ref{section:dynamical-systems}, we will show how a solution of Problem \ref{prob:main} allows us to efficiently compute approximations for $\prob_{x_{t}}$ with formal guarantees of correctness in the $\rho$-Wasserstein metric.
\end{example}

We should also stress that the impact of a solution to Problem \ref{prob:main} is not limited to dynamical systems and, for instance, also represents a key contribution to the distributional robust uncertainty propagation quantification problem, where it is still an open question how to quantify $\sup_{\probQ \in \mathbb{B}_\theta(\prob)} \wasserstein_{\rho}(f \# \probQ, f \# \prob)$ when $f$ is non-linear
\cite{aolaritei2022distributional}. A solution to Problem \ref{prob:main} would give an efficient method to over-approximate this quantity\footnote{Indeed, as a corollary of \eqref{eq:ProbConverg} and \eqref{Eqn:probBounds}, it holds that \(\sup_{\probQ \in \mathbb{B}_\theta(\prob)} \wasserstein_{\rho}(f \# \probQ, f  \# \prob) \leq \boundProbDef + \epsilon\).}.

\begin{remark}[Wasserstein distance vs. divergences]
    A key advantage in using the $\rho$-Wasserstein distance to quantify the error compared to other commonly used quantities, such as KL divergence \cite{gibbs2002choosing}, is that bounds in the  $\rho$-Wasserstein distance between two probability distributions can be used to bound their difference in moments (\cite{adams2024finite}, Lemma 2), in probability (\cite{gao2023distributionally}, Example 7), and many other further quantities of interest (\cite{ernst2022wasserstein}, Section 4). 
\end{remark}

\textbf{Approach.}
In Section \ref{section:wasserstein-bounds}, we start by focusing on bounding $\sup_{\probQ \in \mathbb{B}_\theta(\prob)} \wasserstein_{\rho}(f \# \probQ, f \# \Delta_{\bsR, \bsC} \# \prob)$ for a given ${\bsR, \bsC}$ using results from stochastic optimization and properties of the Wasserstein distance and derive bounds both for $\theta>0$ and $\theta=0$. Then, in Section \ref{section:optimize-locs-convergence}, we present an algorithm to efficiently select the partition $\bsR$ and locations $\bsC$, and we further prove the convergence of $\sup_{\probQ \in \mathbb{B}_\theta(\prob)} \wasserstein_{\rho}(f \# \probQ, f \# \Delta_{\bsR, \bsC} \# \prob)$ to  $\sup_{\probQ \in \mathbb{B}_\theta(\prob)} \wasserstein_{\rho}(f \# \probQ, f \# \prob)$ for the resulting algorithm as the number of locations $|\bsC|$ increases. Lastly, in Section \ref{section:dynamical-systems}, we show how our uncertainty propagation framework can be applied to approximate the state distributions in stochastic dynamical systems with formal $\rho$-Wasserstein guarantees for both finite and infinite prediction horizons. Section \ref{section:experimental-results} provides experimental results on various benchmarks to show the effectiveness of our approach. 

\section{Error Bounds in Wasserstein Distance}\label{section:wasserstein-bounds}
In this section, we show how for a given quantization operator $\Delta_{\bsR, \bsC}$ one can efficiently bound $\sup_{\probQ \in \mathbb{B}_\theta(\prob)} \wasserstein_{\rho}(f \# \probQ, f \# \Delta_{\bsR, \bsC} \# \prob)$ for any $\theta\geq 0$. Our main result is reported next and is based on a norm linearization around each location $c_k \in \bsC$. 
\begin{theorem}[Uncertainty propagation of ambiguity sets]
    \label{th:bound}
     For $\sX \subseteq \realNum^d$ and $\prob \in \sP_\rho (\sX)$, assume a given $\sX$-partition $\bsR = \big\{ \sR_k \big\}_{k=1}^N$ and a set of locations $\bsC = \big\{ {c}_k \big\}_{k=1}^N$. For every $k\in\{1,...,N\}$, call $\evpi{k} \coloneqq \prob(\sR_k)$, and let $\alpha_k, \beta_k \in \realNum_{+}$ be such that for $x\in \sX$
    \begin{equation}\label{eq:norm-linearization}
        \norm{f(x) - f(\vc_k)}^\rho \leq \alpha_k \norm{x - \vc_k}^\rho + \beta_k.
    \end{equation}  
    Further, denote
     \begin{equation}\label{def:theta_d-definition}
         \theta_d = \bigg( \sum_{k=1}^N \int_{\sR_k} \norm{x-c_k}^\rho d\prob(x) \bigg)^\frac{1}{\rho}
     \end{equation}
    Then, for $\alpha_{\max} = \max_{k \in \{1,...,N\}} \alpha_k$, it holds that
    \begin{align}\label{eq:main-result-bound}
        \maybeamp\sup_{\probQ \in \mathbb{B}_\theta(\prob)} \wasserstein_\rho(f\#\probQ, f\#\Delta_{\bsR, \bsC}\#\prob) \maybenewline\maybeamp\maybeqquad\leq \bigg( \alpha_{\max} (\theta+\theta_d)^\rho + \sum_{k=1}^{N} \evpi{k}\beta_k \bigg)^\frac{1}{\rho}.\maybenonumber
    \end{align}
\end{theorem}
The proof of Theorem~\ref{th:bound} is given in Appendix~\ref{proof:th:bound}, where we rely on duality to relax the computation of $ \sup_{\probQ \in \mathbb{B}_\theta(\prob)} \wasserstein_\rho(f\#\probQ, f\#\Delta_{\bsR, \bsC}\#\prob)$ into a one-dimensional minimization problem that can be efficiently bounded by using the fact that $\wasserstein_\rho(\prob, \#\Delta_{\bsR, \bsC}\#\prob)$ can be formulated as a semi-discrete optimal transport problem (Proposition \ref{prop:compute-theta-d}), and on the local linearization of $f$ given in \eqref{eq:norm-linearization}. An algorithm to automatically select $\alpha_k$ and $\beta_k$ for the various locations will be given in Section \ref{subsection:approx-algo}, while how to compute $\theta_d$ has already been mentioned in Remark \ref{remark:theta-d-how-to-compute}.
Before discussing the theoretical implications of Theorem \ref{th:bound} in the rest of this section, we should stress that a potential source of conservatism in Theorem \ref{th:bound} is in the linearization around each location $c_k$, which must hold for all $x\in \sX$ and not just locally in $\sR_k$. This is due to the uncertainty of not knowing $\prob$ exactly. %In section \ref{sec:theta0} we then show how the resulting bound can be substantially improved for $\theta=0$. 
In Subsection \ref{sec:theta0}, we show that such a requirement can be relaxed, and consequently, the bound is improved when there is no ambiguity set, i.e., $\theta=0$.

\begin{remark}[Lipschitz-based uncertainty propagation]\label{rem:linearSyst}
Note that a straightforward corollary of Theorem \ref{th:bound} is that
\begin{align}
    \label{eq:LipschitzBound}
   \sup_{\probQ \in \mathbb{B}_\theta(\prob)} \wasserstein_\rho(f\#\probQ, f\#\Delta_{\bsR, \bsC}\#\prob) \leq \sL_f \big( \theta+\theta_d \big),
\end{align}
where $\sL_f$ is the Lipschitz constant of $f$ according to the $L_\rho$-norm\footnote{This follows by  observing that for all $j\in\{1,...,N\}$ one can select $\beta_j=0$  and $\alpha_j=\sL_f$. The resulting choice always satisfies \eqref{eq:norm-linearization} by the definition of Lipschitz constant.}.
However, in general, as we give intuition in Example \ref{exampl:ImportanceOdAlphaBeta} below, and we will show empirically in Section \ref{section:experimental-results}, the bound in \eqref{eq:main-result-bound} is generally substantially tighter and can return bounds that are orders of magnitude smaller. The intuition is that in the regions $\mathcal{R}_i$ where the local Lipschitz constant of $f$ is large, one can rely on a larger $\beta_j$ to obtain a lower $\alpha_{\max}$. If, in these regions, the probability mass of $\prob$ is small (and, consequently, $\evpi{j}$ is low), then the bound could substantially improve.  
Note that an exception is when $f$ is linear, where it is easy to show that the bounds in \eqref{eq:main-result-bound} and \eqref{eq:LipschitzBound} coincide.
In fact,  if  $f$ is  linear, i.e., $f(x) \coloneqq Ax + b$, Theorem \ref{th:bound} reduces to
   \begin{align}
   \label{eqn:LinearCase}
   \sup_{\probQ \in \mathbb{B}_\theta(\prob)} \wasserstein_\rho(f\#\probQ, f\#\Delta_{\bsR, \bsC}\#\prob) \leq \norm{A} \big( \theta+\theta_d \big),
   \end{align}
   where $\norm{A} \coloneqq \sup_{x \in \sX} \frac{\norm{Ax}}{\norm{x}}$ is the induced norm of $A$, which is equivalent to the global Lipschitz constant of $f$. %One can further show (see Proposition \ref{prop:remark-linear-case} in Appendix \ref{appendix:additional-results}) that in this case $\sup_{\probQ \in \mathbb{B}_\theta(\prob)} \wasserstein_\rho(f\#\probQ, f\#\Delta_{\bsR, \bsC}\#\prob) \in \big[ \norm{A}(\theta-\theta_d), \norm{A}(\theta+\theta_d) \big]$.
\end{remark}

\begin{figure}[htbp]
    \centering
    \ifthenelse{\boolean{doublecolumn}}{
        \includegraphics[width=0.8\linewidth]{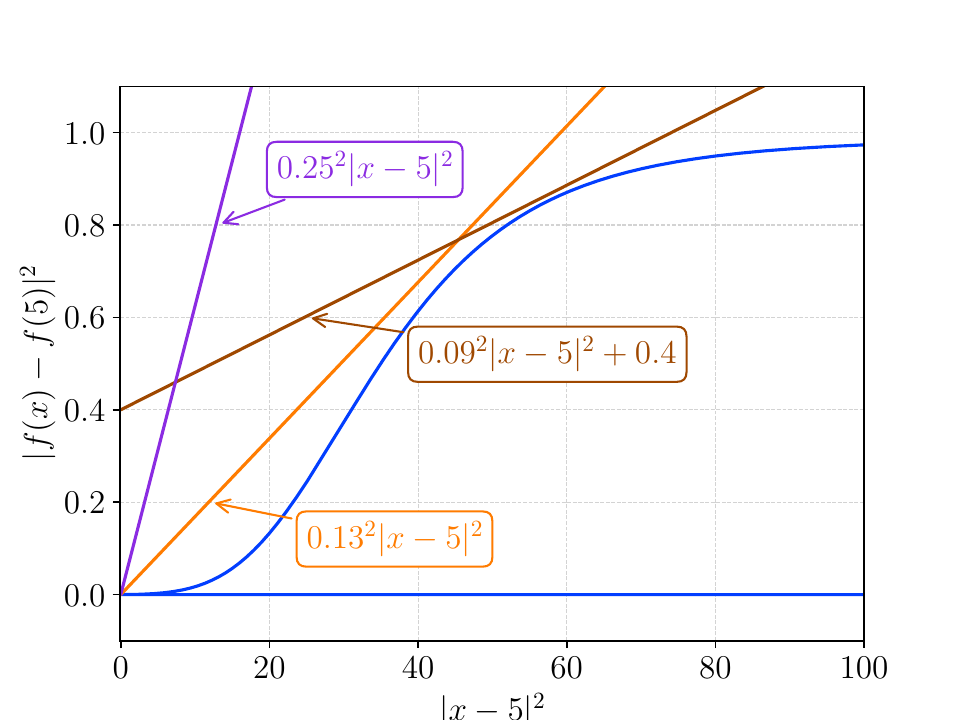}
    }{
        \includegraphics[width=0.6\linewidth]{Figures/Examples/comparison_approx_sigmoid.pdf}
    }
    \caption{There exists infinite admissible pairs $(\alpha, \beta)$ such that \eqref{eq:norm-linearization} holds. In particular, we show three of them: $(0.25^2, 0)$ (purple line), $(0.13^2, 0)$ (orange), and $(0.09^2, 0.4)$ (brown).}
    \label{fig:norm-linearization-representation}
\end{figure}

\begin{example}[Local vs. Lipschitz-based norm approximations]
\label{exampl:ImportanceOdAlphaBeta}
    Let $\rho=2$, and $f(x)= \frac{1}{1+e^{-x}}$ be a sigmoid function, whose Lipschitz constant w.r.t. the $L_2$-norm is $\sL_f=0.25$. We consider the location $c=5$; see Figure \ref{fig:norm-linearization-representation}. We can upper bound $|f(x)-f(5)|^2$ for any $x\in \realNum$ with \eqref{eq:norm-linearization} by choosing i) $(\alpha, \beta)=(0.25^2, 0)$, ii) $(\alpha, \beta)=(0.13^2, 0)$, or iii) $(\alpha, \beta)=(0.09^2, 0.4)$. The first observation is that since the location $c=5$ is far from the region where the global Lipschitz is found ($x=0$), $\alpha$ can be chosen to be considerably smaller than $\sL_f^2$ even for $\beta=0$ $(0.13^2 < 0.25^2)$. 
    Further, since $f$ is bounded, by increasing the bias $\beta$, one can decrease $\alpha$ even further (see the brown line). %As we can infer from the simple example, there are an infinite number of possible combinations $(\alpha, \beta)$ that lead to a valid overapproximation in Eqn \eqref{th:bound}. 
    In Section \ref{subsection:approx-algo}, we explain how to automatically select these parameters.%, and in Example \ref{example:continuation-part2}, we show its effect on the bounds provided by Theorem \ref{th:bound} in a practical example.
\end{example}

\subsection{No ambiguity case: $\theta = 0$.}
\label{sec:theta0}
As mentioned, when $\theta = 0$, the bound in Theorem \ref{th:bound} can be improved. In fact, in the proof of Theorem \ref{th:bound}, as will be discussed in detail in Section \ref{remark:ImprovedBoundTheta0}, to obtain a tractable reformulation, we seek the worst plausible joint distribution among all couplings such that one of the marginals is $\Delta_{\bsR, \bsC}\#\prob$ and the other is any distribution $\probQ \in \mathbb{B}_{\theta+\theta_d}(\prob)$. Instead, when $\prob$ is known, as in the case $\theta=0$, we can design a specific transport plan that generally leads to improved bounds, as shown in Theorem \ref{th:bound-zero-ball} below. 
\begin{theorem}[Uncertainty propagation under no ambiguity]\label{th:bound-zero-ball}
     For $\sX \subseteq \realNum^d$, let $\prob \in \sP_\rho (\sX)$. Assume a given partition $\bsR = \big\{ \sR_k \big\}_{k=1}^N$ and a set of locations $\bsC = \big\{ {c}_k \big\}_{k=1}^N$. For every $k\in\{1,...,N\}$, call $\evpi{k} \coloneqq \prob(\sR_k)$, and let $\alpha_k, \beta_k \in \realNum_{+}$ be such that for $x\in \sR_k$, it holds that
    \begin{align}\label{eq:norm-inequality-zero-budget}
        \norm{f(x) - f(\vc_k)}^\rho \leq \alpha_k \norm{x - \vc_k}^\rho + \beta_k.
    \end{align}
    Then,
    \begin{align}
        \maybeamp\wasserstein_\rho(f\#\prob, f\#\Delta_{\bsR, \bsC}\#\prob) \maybenonumber \maybenewline\maybeamp\maybeqquad \leq  \left( \sum_{k=1}^N \alpha_k \int_{\sR_k} \norm{x-c_k}^\rho d\prob(x) + \sum_{k=1}^N \evpi{k} \beta_k  \right)^\frac{1}{\rho} \label{eq:main-result-bound-zero-ball}
    \end{align}
\end{theorem}
The proof of Theorem \ref{th:bound-zero-ball} is reported in Appendix \ref{proof:th:bound-zero-ball}. Note that, differently from Theorem \ref{th:bound}, the norm overapproximation in \eqref{eq:norm-inequality-zero-budget} is local, i.e. for each region $\sR_k$, \eqref{eq:norm-inequality-zero-budget} has to hold for every $x \in \sR_k$, instead of $x \in \sX$ as in Theorem \ref{th:bound}. This is a consequence of the fact that, in the setting of Theorem \ref{th:bound-zero-ball}, $\prob$ is known with no uncertainty.

\subsection{Conservatism under ambiguity}\label{remark:ImprovedBoundTheta0}
We should stress that, although $\sup_{\probQ \in \mathbb{B}_\theta(\prob)} \wasserstein_\rho(f\#\probQ, f\#\Delta_{\bsR, \bsC}\#\prob)$ is right-continuous in $\theta=0$, the bound in Theorem \ref{th:bound} does not generally converge to the one in \ref{th:bound-zero-ball} as $\theta \downarrow 0$. 
To see this, note that the proof of Theorem~\ref{th:bound} is based on a worst-case analysis. In particular, as detailed in Appendix~\ref{proof:th:bound}, we define $ S_\theta(\mathbb{T}) := \Big\{ \gamma \in \sP(\sX \times \sX) : \int_{\sX \times \sX} \norm{x_1 - x_2}^\rho d\gamma(x_1, x_2) \leq \theta^\rho, \text{proj}_2 \# \gamma = \mathbb{T} \Big\}$, and show
\begin{align*}
    \maybeamp\sup_{\probQ \in \mathbb{B}_{\theta}(\prob)}\wasserstein_\rho(f\#\probQ, f\#\Delta_{\bsR, \bsC}\#\prob) \maybenewline\maybeamp\maybeqquad \leq \sup_{\gamma \in S_{\theta+\theta_d}(\Delta_{\bsR, \bsC}\#\prob)} \int_{\sX \times \sX} \norm{f(x)-f(y)}^\rho d\gamma(x, y)
\end{align*}
By taking the limit $\theta \downarrow 0$ on both sides, we have
\begin{align}
    \maybeamp\wasserstein_\rho(f\#\prob, f\#\Delta_{\bsR, \bsC}\#\prob)\label{eq:bound-analysis-thm4} \maybenewline\maybeamp\maybeqquad\leq \sup_{\gamma \in S_{\theta_d}(\Delta_{\bsR, \bsC}\#\prob)} \int_{\sX \times \sX} \norm{f(x)-f(y)}^\rho d\gamma(x, y) \maybenonumber  
\end{align}
On the other hand, to prove Theorem~\ref{th:bound-zero-ball}, 
we design a specific coupling $\gamma^* \in \Gamma(\prob, \Delta_{\bsR, \bsC}\#\prob)$ that achieves $\int_{\sX \times \sX} \norm{x-y}^\rho d\gamma^*(x, y) = \theta_d^\rho$, as reported in ~\eqref{eq:general-coupling} in the Appendix~\ref{proof:th:bound}, and bound
\begin{align}
    \maybeamp\wasserstein_\rho(f\#\prob, f\#\Delta_{\bsR, \bsC}\#\prob) \label{eq:bound-analysis-thm6}\maybenewline 
    \maybeamp\maybeqquad\leq \int_{\sX \times \sX} \norm{f(x)-f(y)}^\rho d\gamma^*(x, y) \maybenonumber
\end{align}
While it holds by construction that $\gamma^*$ is in $S_{\theta_d}(\Delta_{\bsR, \bsC}\#\prob)$, the elements in \(S_{\theta_d}(\Delta_{\bsR, \bsC}\#\prob)\) do not necessarily have $\prob$ as one of the marginals, i.e., are not necessarily a member of \(\Gamma(\prob, \Delta_{\bsR, \bsC}\#\prob)\). Hence, in general, the bounds in ~\eqref{eq:bound-analysis-thm4} and that in
~\eqref{eq:bound-analysis-thm6} for $\theta=0$ differ
\footnote{For instance, consider $\prob = \delta_{(0,0)}$, $\bsR \coloneqq \{ \realNum^2 \}$, $\bsC \coloneqq \{ (0, \theta_d) \}$, and $f(x) \coloneqq \text{diag}(2, 0.1)x$. Then, $\Delta_{\bsR, \bsC}\#\prob = \delta_{(0, \theta_d)}$, and $\gamma^*=\delta_{(0,0)\times(0,\theta_d)}$. Note that $\Tilde{\gamma} \coloneqq \delta_{(0, \theta_d)\times(\theta_d, \theta_d)} \in S_{\theta_d}(\delta_{(0, \theta_d)})$, and hence, $\int_{\sX \times \sX} \norm{f(x)-f(y)}^\rho d\gamma^*(x, y) = 0.1^\rho \theta_d^\rho$, while $\int_{\sX \times \sX} \norm{f(x)-f(y)}^\rho d\Tilde{\gamma}(x, y) = 2^\rho \theta_d^\rho$, a significantly larger value.}.

\section{Selecting approximate discrete distributions}\label{section:optimize-locs-convergence}
In this section, we provide an algorithmic approach to automatically select $\bsR$, $\bsC$, and the linearization coefficients in Theorems \ref{th:bound} and \ref{th:bound-zero-ball}. 
% to minimize the approximation error bound in Eqn \eqref{eq:main-result-bound} or \eqref{eq:main-result-bound-zero-ball}. 
First, in Section \ref{subsection:approx-algo}, for any $c \in \sX \subseteq \realNum^d$, we present an algorithm to compute coefficient pairs $(\alpha, \beta)$ such that either \eqref{eq:norm-linearization} or \eqref{eq:norm-inequality-zero-budget} holds, and the corresponding error bound in \eqref{eq:main-result-bound} or \eqref{eq:main-result-bound-zero-ball} is minimized.
%, and we explain how, for given partition $\bsR$ and locations $\bsC$, it can be used to compute the $\rho$-Wasserstein bounds proposed in Theorems \ref{th:bound} and \ref{th:bound-zero-ball}.
Then, we provide a practical approach to construct a partition $\bsR$ and a set of locations $\bsC$ that guarantees that the approximation error defined in \eqref{eq:ProbConverg} can be made arbitrarily small.

\subsection{Norm approximation algorithm}\label{subsection:approx-algo}
As observed in Example \ref{exampl:ImportanceOdAlphaBeta}, given a location $c \in \bsC$, there exist infinite combinations of $(\alpha, \beta)$ to generate the upper-bounds for $\norm{f(x)-f(c)}^\rho$ for all $x \in \sX$. 
Unfortunately, due to the possibly non-linear nature of \(f\), it is generally intractable to minimize the error bound in Theorem \ref{th:bound} or \ref{th:bound-zero-ball} with respect to all feasible linearization combinations. 
Hence, in practice, we 
focus on combinations of type (i) \((0, \beta)\) and type (ii) \((\alpha, 0)\), which can be computed efficiently.
Specifically, for combinations of type (i), where $f$ remains bounded in the region where the linear bounds must hold, we select $(\alpha, \beta) = (0,\sup_{x\in\sX} \norm{f(x)-f(c)}^\rho)$.
% \footnote{For Theorem \ref{th:bound-zero-ball}, the norm-linearization only needs to hold over a region \(\sR\subseteq\sX\), thus \((0,\sup_{x\in\sR} \norm{f(x)-f(c)}^\rho)\) can be selected for type ii). Note, however, that computing the local sup generally involves solving a non-convex optimization problem, and approximate solutions using techniques such as bound propagation are required.} 
For type (ii) combinations, we select \((\alpha,\beta)=(\sup_{x\in \sX} \norm{f(x)-f(c)}^\rho/\norm{x-c}^\rho, 0)\). Due to the typically non-convex nature of these optimization problems, in practice, we utilize approximate solutions obtained via bound propagation techniques.\footnote{For our experiments, we use the linear bound propagation techniques from \cite{mathiesen2022safety} to compute the linear maps.} 
That is, for each region \(\sS_k\) of a $\sX$-partition $\bsS$, we compute 
linear maps \(\check{\mA_k}(\vx-c)\) and \(\hat{\mA}_k(\vx-c)\), 
and vectors \(\check{\vb}_k\) and \(\hat{\vb}_k\),
that satisfy: 
\begin{gather}
    \check{\mA}_k(\vx-c) \preceq f(\vx) - f(c) \preceq \hat{\mA}_k(\vx-c)\label{eq:linear-maps-condition}\\
    \check{\vb}\preceq f(\vx) - f(c) \preceq \hat{\vb}, 
\end{gather}
for all \(\vx\in\sS_k\). We then use that
\begin{equation}\label{eq:heuristics-for-alpha}
    \sup_{x\in \sX}\frac{\norm{f(x)-f(c)}^\rho}{\norm{x-c}^\rho} \leq 
    \max_{k \in \{1,...,N\}}\left(\|\check{\mA}_k\|^\rho,\|\hat{\mA}_k\|^\rho\right).
\end{equation}
and
\begin{equation}\label{eq:heuristics-for-beta}
    \sup_{x\in\sX} \norm{f(x)-f(c)}^\rho \leq \max_{k \in \{1,...,N\}}\left(\|\check{\vb}_k\|^\rho,\|\hat{\vb}_k\|^\rho\right)
\end{equation} 
to set \(\alpha\) for combinations of type (ii), and \(\beta\) for combinations of type (i) and respectively. Note that, for Theorem \ref{th:bound-zero-ball} where the norm-linearization has to hold only over a region \(\sR\subseteq\sX\), we follow the same procedure, replacing \(\sX\) by \(\sR\).

Algorithm~\ref{alg:compute-bound} details a procedure to select $(\alpha_k,0)$ or $(0,\beta_k)$ for all $c_k \in \bsC$ for Theorem~\ref{th:bound}. The case for Theorem \ref{th:bound-zero-ball} follows similarly. Algorithm~\ref{alg:compute-bound} is based on the fact that Theorem~\ref{th:bound} only depends on the maximum value of the $\alpha_k$ coefficients for all locations $c_k \in \bsC$, so, by ordering those coefficients in descending order, we can iteratively verify whether replacing $\alpha_k\norm{x-c_k}^\rho$ approximations in ~\eqref{eq:norm-linearization} for $\beta_k$ lead to a tighter bound. As discussed in Remark~\ref{rem:linearSyst}, this will generally be the case when $\prob(\sR_k)$ is low.   
More specifically, we start by computing $\bar\alpha \coloneqq (\alpha_1,...,\alpha_N)$ and $\bar\beta \coloneqq (\beta_1,...,\beta_N)$ using \eqref{eq:heuristics-for-alpha} and \eqref{eq:heuristics-for-beta}, respectively, for each $c_k \in \bsC$ (line 2), and $\theta_d$ in \eqref{def:theta_d-definition} as explained in Remark \ref{remark:theta-d-how-to-compute} (line 3). We then compute the first candidate for the bound, by applying Theorem \ref{th:bound} with $\bar\alpha$ (line 4). In line 5, we sort $\bar\alpha$ in descending order (and sort accordingly $\bar\beta, \bsR, \bsC$). As mentioned above, the strategy is to try to replace the highest $\alpha_k$ by zero, and include instead $\beta_k$, while verifying if the bound decreases. This is implemented in the \emph{for} loop in lines 6-12. %In general, the bounds will decrease when $\prob(\sR_k)$ is sufficiently low.
\RestyleAlgo{ruled}
\SetKwComment{Comment}{/* }{ */}
\begin{algorithm}[h]
\DontPrintSemicolon
\caption{Compute least conservative bound in Thm \ref{th:bound} for a given quantization operator
}\label{alg:compute-bound}
\KwIn{$\sX$-partition $\bsR$, set of locations $\bsC$, radius $\theta$, distribution $\prob$}
\KwOut{Least conservative $\rho$-Wasserstein bound in Thm \ref{th:bound} given $\bsR, \bsC$}
\SetKwFunction{FMain}{BoundGivenQuantizationOperator}
\SetKwProg{Fn}{function}{:}{}
\Fn{\FMain{$\bsR, \bsC, \theta, \prob$}}{
    % $\bar{\alpha} \gets (\texttt{Lemma \ref{prop:main-heuristics-for-alpha-beta} for each }c_k \in \bsC)_{k=1}^{|\bsC|}$\;
    $(\bar{\alpha},\bar{\beta}) \gets (\texttt{Eqns \eqref{eq:heuristics-for-alpha} \& \eqref{eq:heuristics-for-beta} for }c_k \in \bsC)_{k=1}^{|\bsC|}$\;
    $\theta_d \gets \texttt{Eqn }\eqref{def:theta_d-definition}$\;
    $\wasserstein \gets \max_{\alpha \in \bar{\alpha}  } \alpha (\theta + \theta_d)$\;
    $\bar{\alpha}_\text{sorted}, \bar{\beta}_\text{sorted}, \bsR_\text{sorted}, \bsC_\text{sorted} \gets \texttt{sort descendingly according to }\bar{\alpha}$\;
    \For{$k \in \{1,...,|\bsC|\}$}{
        % $b_k \gets \sum_{j=1}^{k} \prob(\sR_{\text{sorted}, j}) \sup_{x \in \sX} \norm{f(x)-f(c_j)}^\rho$\;
        $b_k \gets \sum_{j=1}^{k} \prob(\sR_{\text{sorted}, j}) \bar{\beta}_{\text{sorted},j}$\;
        $\Tilde{\wasserstein} \gets \bigg(\bar{\alpha}_{\text{sorted},k+1} (\theta + \theta_d)^\rho + b_k \bigg)^\frac{1}{\rho}$\;
        \If{
            $\wasserstein\leq\Tilde{\wasserstein}$
        }{
            \texttt{break}\;}
            {
        \Else{
            $\wasserstein \gets \Tilde{\wasserstein}$\;
        }}
        % $\wasserstein \gets \min \{\wasserstein, \Tilde{\wasserstein} \}$\;
    }
    \Return{$\wasserstein$}
}
\end{algorithm}
\subsection{Constructing a converging quantization operator}\label{section:constructing-converging-quant-operator}
After having discussed how to select the linearization coefficients in Theorem \ref{th:bound} and \ref{th:bound-zero-ball}, what is left to do is to explain how to effectively construct a quantization operator $\Delta_{\bsR, \bsC}$, i.e. a $\sX$-partition $\bsR$ and a set of locations $\bsC \subset \sX$, such that the convergence requirement in Problem \ref{prob:main} holds for any given $\epsilon >0$. 
We start with the following lemma, which is a straightforward consequence of the triangular inequality, showing that to satisfy Problem \ref{prob:main}, it is enough to select $\Delta_{\bsR, \bsC}$ to minimize $\wasserstein_{\rho}(f  \# \prob, f \# \Delta_{\bsR, \bsC} \# \prob).$ This result implies that to guarantee an arbitrarily small \(\epsilon\), it is enough to optimize $\bsR$, $\bsC$ w.r.t. to $\prob$ even if $\theta>0$.

\begin{lemma}\label{prop:bound-on-convergence-problem-1}
    Let $\prob \in \sP_\rho(\sX)$. For any $\sX$-partition $\bsR$, and set of locations $\bsC$, it holds that
    \begin{align}
        \maybeamp\Bigg| \sup_{\probQ \in \mathbb{B}_\theta(\prob)} \wasserstein_{\rho}(f \# \probQ, f \# \Delta_{\bsR, \bsC} \# \prob) \label{eq:bound-approx-error-triangle-ineq}\maybenewline\maybeamp\maybeqquad - \sup_{\probQ \in \mathbb{B}_\theta(\prob)} \wasserstein_{\rho}(f \# \probQ, f  \# \prob) \Bigg| \leq \wasserstein_{\rho}(f  \# \prob, f \# \Delta_{\bsR, \bsC} \# \prob) \maybenonumber
    \end{align}
\end{lemma}

Lemma \ref{prop:bound-on-convergence-problem-1} is used in the next theorem to show that even taking  $\bsR$ as a uniform partitioning of any compact set containing enough probability mass of $\prob$ to select $\bsR, \bsC$ can guarantee a solution to Problem \ref{prob:main}. An improved, non-uniform, partitioning scheme will then be given in Remark \ref{remark:algorithm-quantization-gaussian-cas}.
\begin{theorem}[Convergence rate]
% quantization continuity of the worst-case uncertainty propagation
\label{prop:convergence-algorithm-error-approx}
    For $\sX \subseteq \realNum^d$, let $\prob \in \sP_\rho(\sX)$ and $\rho\geq 1$. For any $\epsilon > 0$, consider a cubic compact set $\bar\sX \subseteq \sX$ such that $\int_{\sX \setminus \bar\sX} \norm{x - \bar{c}}^\rho d\prob(x) \leq \frac{\epsilon^\rho}{2 \sL_f^\rho}$ for some $\bar{c} \in \sX$. Further, consider $\bsR \coloneqq \{ \sR_k \}_{k=1}^N$ as a uniform $\bar\sX$-partition of $\bar\sX$ in $N \geq \left( \frac{2^\frac{1}{\rho} \sL_f d^\frac{1}{\rho} \norm{\bar\sX}_\infty}{\epsilon} \right)^d$ hypercubic regions, and $\bsC$ as set of the centers of each hypercube $\sR_k$. 
    Then, for $\bsR^* = \bsR \cup \{ \sX \setminus \bar\sX \}$ and $\bsC^* = \bsC \cup \{ \bar{c} \}$, it holds that:
    \begin{align}
       \maybeamp\Bigg| \sup_{\probQ \in \mathbb{B}_\theta(\prob)} \wasserstein_{\rho}(f \# \probQ, f \# \Delta_{\bsR^*, \bsC^*} \# \prob) \maybenonumber \maybenewline\maybeamp\maybeqquad\maybeqquad\maybeqquad - \sup_{\probQ \in \mathbb{B}_\theta(\prob)} \wasserstein_{\rho}(f \# \probQ, f  \# \prob) \Bigg| \leq \epsilon.
    \end{align}
\end{theorem}
The convergence rate reported in Theorem  \ref{prop:convergence-algorithm-error-approx} is conservative due to two factors: i) 
it relies on linearization coefficients \((\alpha,\beta)=(\lipschitz_f,0)\) in Theorem \ref{th:bound-zero-ball}, which generally leads in over-conservative error bounds, as discussed in Remark~\ref{rem:linearSyst}, ii) Theorem \ref{prop:convergence-algorithm-error-approx} is proven w.r.t. a uniform partitioning of an appropriately selected compact set. 
Consequently, it is evident that non-uniform partitioning approaches that directly minimize the bounds in Theorem \ref{th:bound-zero-ball} would lead to improved bounds. In particular, we can rely on Lemma \ref{prop:bound-on-convergence-problem-1}, which implies that the quantization error is bounded by
\begin{align}
    \maybeamp\wasserstein_\rho(f\#\prob, f\#\Delta_{\bsR, \bsC}\#\prob) \label{def:h-function}\maybenewline\maybeamp\leq \underbrace{ \left( \sum_{k=1}^N \alpha_k \int_{\sR_k} \norm{x-c_k}^\rho d\prob(x) + \sum_{k=1}^N \evpi{k}\beta_k \right)^\frac{1}{\rho} }_{=: \boundProbDef} \leq \sL_f \theta_d,  \maybenonumber
\end{align}
where $\boundProbDef$ is the error bound from Theorem~\ref{th:bound-zero-ball}.
Consequently, by selecting $\bsR$ and $\bsC$ to minimize $\theta_d$, we indirectly reduce the error bounds. This approach may lead to greatly improved bounds compared to a uniform partitioning approach, as we will illustrate empirically in Section \ref{section:experimental-results}.
\begin{remark}[Partitioning for normalizing flows of Gaussians]\label{remark:algorithm-quantization-gaussian-cas}
    When \(\prob\) is Gaussian or a normalizing flow of a latent Gaussian distribution \cite{papamakarios2021normalizing},\footnote{For normalizing Gaussian distribution flows, i.e. $\prob \coloneqq g \# \sN(\mu, \Sigma)$ for some known piecewise Lipschitz continuous function $g$, one can use that $\wasserstein_\rho(g\#\sN(\mu, \Sigma), g\#\Delta_{\bsR, \bsC}\#\sN(\mu, \Sigma)) \leq \sL_g \wasserstein_\rho(\sN(\mu, \Sigma), \Delta_{\bsR, \bsC}\#\sN(\mu, \Sigma))$.} 
    we can rely on Algorithm 2 from \cite{adams2024finite} 
    to obtain $\bsC$ that minimize $\theta_d$, with \(\bsR\) being the Voronoi partition w.r.t. $\bsC$.
    Since Algorithm 2 from \cite{adams2024finite} guarantees that $\theta_d$ converges to zero as $N$ increases, we can iteratively increase the number of locations \(N\) until $\lipschitz_f\theta_d\leq\epsilon$, and consequently, according to ~\eqref{def:h-function}, $\boundProbDef \leq \epsilon$, where \(\epsilon>0\) is the desired error threshold. 
    The non-uniform partition \(\bsR\) resulting from Algorithm 2 in \cite{adams2024finite} typically leads to a convergence rate that is significantly better than the one presented in Theorem \ref{prop:convergence-algorithm-error-approx}.
\end{remark}

\begin{figure}[ht]
    \centering
    \includegraphics[width=1.0\linewidth]{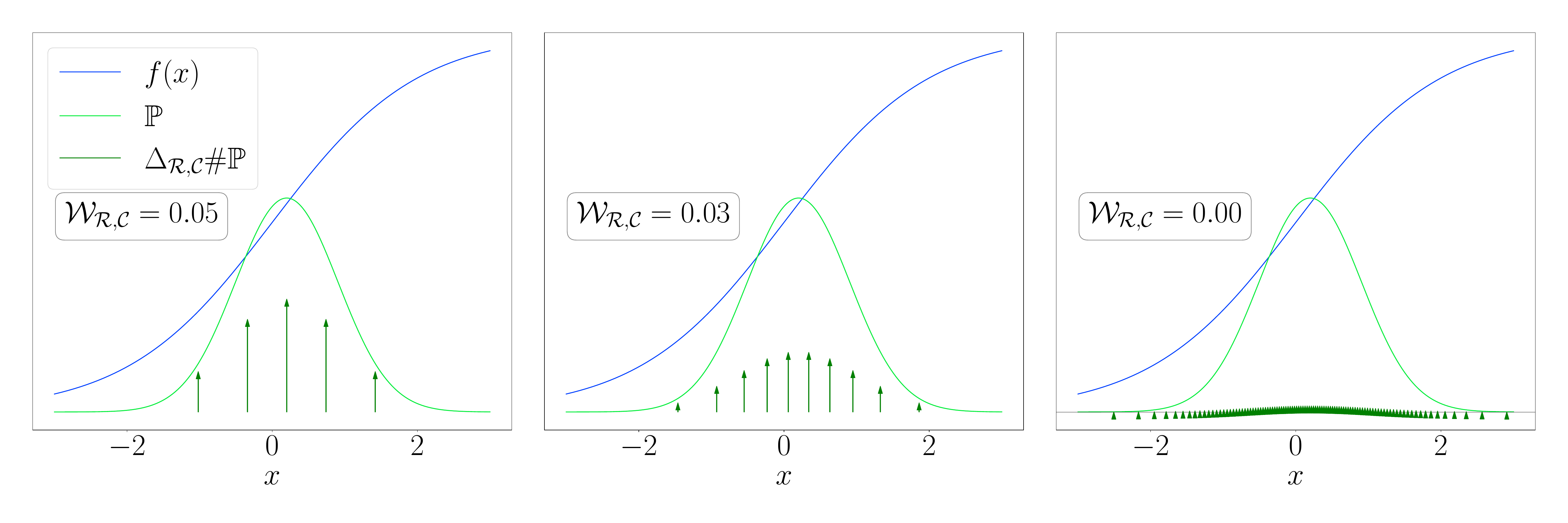}
    \caption{
    Quantization of $\prob = \sN(0.2, 0.5)$, $\Delta_{\bsR, \bsC}\#\prob$, constructed as described in Remark \ref{remark:algorithm-quantization-gaussian-cas}, for $|\bsC| \in \{ 5, 10, 10^2 \}$, and the corresponding bound $\boundProbDef$ for $f$ the sigmoid function.}
    \label{fig:example3}
\end{figure}

\begin{example}[Efficacy of Algorithm \ref{alg:compute-bound}]\label{example:continuation-part2}
    Let $\rho=2$. Consider again the sigmoid function $f:\realNum \rightarrow \realNum$ of Example \ref{exampl:ImportanceOdAlphaBeta}.
    % given by $f(x) \coloneqq \frac{1}{1+e^{-x}}$. 
    Further, let $\prob = \sN(0.2, 0.5)$. Figure \ref{fig:example3}, illustrates $\Delta_{\bsR, \bsC}\#\prob$ and shows $\boundProbDef$ for different number of location $|\bsC|$. 
    Note how the bound monotonically decreases and reaches a value in the order of $10^{-2}$ with only $10$ locations.
        % $N \in \{ 5, 10, 10^2 \}$, we show in Figure \ref{fig:example3}  
    % $\Delta_{\bsR, \bsC}\#\prob$ obtained with the algorithm above for $N \in \{ 5, 10, 10^2 \}$, and compute $h(\bsR, \bsC)$. 
 \end{example}

% \ifthenelse{\boolean{doublecolumn}}{}{\newpage}
\section{Iterative predictions for stochastic dynamical systems}\label{section:dynamical-systems}
In this section, we show how our results can be used to generate provably correct discrete approximations for stochastic dynamical systems with formal guarantees in the $\rho$-Wasserstein distance. To do so, we consider the general model of a discrete-time stochastic process already introduced in Example \ref{example:initial}:
\begin{equation}\label{eq:system-def}
    x_{t+1} = f(x_t, \omega_t), \quad x_0 \sim \prob_{x_0}, \omega_t \sim \prob_\omega,
\end{equation}
where the measurable function $f:\sX \times \sW \rightarrow \sX$, 
with $\sX \subseteq \realNum^d$ as the state space and $\sW \subseteq \realNum^q$ as uncertainty space, 
represents the one-step dynamics of the system. Here, $x_0$ is the initial condition of the system, assumed to be distributed with distribution $\prob_{x_0} \in \sP(\sX)$, and $\omega_t$ is an i.i.d. process noise distributed according to $\prob_\omega \in \sP(\sW)$.
We denote the state-noise joint distribution at time \(t\) by $\prob_t \coloneqq \prob_{x_{t}} \times \prob_\omega$.
% , independent of the states $x_t$ at any time step. 
As previously mentioned in Example \ref{example:initial}, if $f$ is non-linear or $\prob_{w}$ non-Gaussian, the distribution $\prob_{x_t}$ of the system at time $t$ becomes intractable. In this Section, we show how our solution of Problem \ref{prob:main} allows one to obtain a tractable (discrete) distribution $\probHat_{x_t} \in \sP(\sX)$ such that $\wasserstein_{\rho}(\prob_{x_t}, \probHat_{x_t}) \leq \delta$, for a given error threshold $\delta > 0$ for any  $t>0$.

Our approach is summarized in Figure \ref{fig:approximation-scheme-dynamical-systems}, where for a time $t$, we denote by  $\bsC_t = \{ c_{t,1}, ..., c_{t, N_t} \}$, $\bsR_t = \{ \sR_{t, 1},...,\sR_{t, N_t} \}$, respectively, the locations and regions for the discrete approximation of the system at time $t$, to emphasize how this can change over time. 
% Similarly, we refer by $(\alpha^t_k, \beta^t_k)$ for the linearization coefficients of Theorem \ref{th:bound} and \ref{th:bound-zero-ball}, which are assumed to be picked as described in Section \ref{subsection:approx-algo}.
To describe our approach, we start with $t=0$, assuming $\prob_{x_0}$ is known, and setting $\probHat_{x_0} = \prob_{x_0}$.
% , which trivially yields $\wasserstein_\rho(\prob_{x_0}, \probHat_{x_0}) = 0$.
For $t=1$, 
% if we denote the true state-noise joint distribution by $\prob_0 \coloneqq \prob_{x_{0}} \times \prob_\omega$, 
the true state distribution is given by $\prob_{x_1} = f\#\prob_0$, which, as we have previously argued, is generally intractable, 
Thus, as showed in Figure \ref{fig:approximation-scheme-dynamical-systems}, we define the approximate state-noise joint distribution as $\probHat_0 = \probHat_{x_{0}} \times \prob_\omega$. 
We then apply the quantization operation using a $(\sX \times \sW)$-partition $\bsR_0$ and a set of locations $\bsC_0 \subset \sX \times \sW$, and propagate it through $f$, resulting in the approximate state (discrete) distribution $\probHat_{x_1} = f\#\Delta_{\bsR_0, \bsC_0}\#\probHat_0$.
% by defining the approximate state-noise joint distribution as $\probHat_0 \coloneqq \probHat_{x_{0}} \times \prob_\omega$, we can first quantize it using a $(\sX \times \sW)$-partition $\bsR_0$ and a set of locations $\bsC_0 \subset \sX \times \sW$, $\Delta_{\bsR_0, \bsC_0}\#\probHat_0$, and then propagate it through $f$, obtaining the approximate state (discrete) distribution $\probHat_{x_1} \coloneqq f\#\Delta_{\bsR_0, \bsC_0}\#\probHat_0$. 
Note that the latter consists of a straightforward application of a $f$ transformation to the support of a discrete distribution, hence providing a tractable propagation through time. This process is repeated for the next time steps, where $(\sX \times \sW)$-partitions $\bsR_t$ and locations $\bsC_t$ are chosen such that the requirement in \eqref{prob:main} is met for some predefined $\epsilon>0$. 

The next result shows how our framework can be applied to bound $\wasserstein_\rho(\prob_{x_t}, \probHat_{x_t})$ for any $t\geq 0$. Furthermore, critically, we show that if $f$ is contractive, the resulting error bounds propagation will reach a fixed point, allowing for infinite-time prediction horizons. 
\begin{figure*}[t]
    \centering
    \begin{tikzcd}
    \textit{Actual state distr.} & \prob_{x_0} \arrow{r} \arrow{d} & \prob_{x_1} \arrow{r} & \prob_{x_2} \arrow{r} & \dots \\
    \textit{Joint distr.} & \underbrace{\prob_{x_0} \times \prob_\omega}_{\probHat_0} \arrow[orange]{d}{\theta_{d,0}} & \underbrace{\probHat_{x_1} \times \prob_\omega}_{\probHat_1} \arrow[orange]{d}{\theta_{d,1}} & \underbrace{\probHat_{x_2} \times \prob_\omega}_{\probHat_2} \arrow[orange]{d}{\theta_{d,2}} & \dots \\
    \textit{Quantization} & \Delta_{\bsR_0, \bsC_0} \# \probHat_0 \arrow{d} & \Delta_{\bsR_1, \bsC_1} \# \probHat_1 \arrow{d} & \Delta_{\bsR_2, \bsC_2} \# \probHat_2 \arrow{d} & \dots \\
    \textit{Approximator} & \underbrace{ f \# \Delta_{\bsR_0, \bsC_0} \# \probHat_0}_{ \probHat_{x_1} } & \underbrace{ f \# \Delta_{\bsR_1, \bsC_1} \# \probHat_1}_{ \probHat_{x_2} } & \underbrace{ f \# \Delta_{\bsR_2, \bsC_2} \# \probHat_2}_{ \probHat_{x_3} } & \dots \\
    \textit{Bounds} & \text{Thm } \ref{th:bound-zero-ball} & \text{Thm } \ref{th:bound} & \text{Thm } \ref{th:bound} & \dots \\
    \arrow[{sloped}, from=4-2, to=2-3]
    \arrow[{sloped}, from=4-3, to=2-4]
    \arrow[{sloped}, from=4-4, to=2-5]
    \end{tikzcd}
    \vspace{-3mm}
    \caption{Discrete approximation scheme for stochastic dynamical systems with formal guarantees on the $\rho$-Wasserstein distance, $\wasserstein_\rho(\prob_{x_t}, \probHat_{x_t})$.}
    \label{fig:approximation-scheme-dynamical-systems}
\end{figure*}
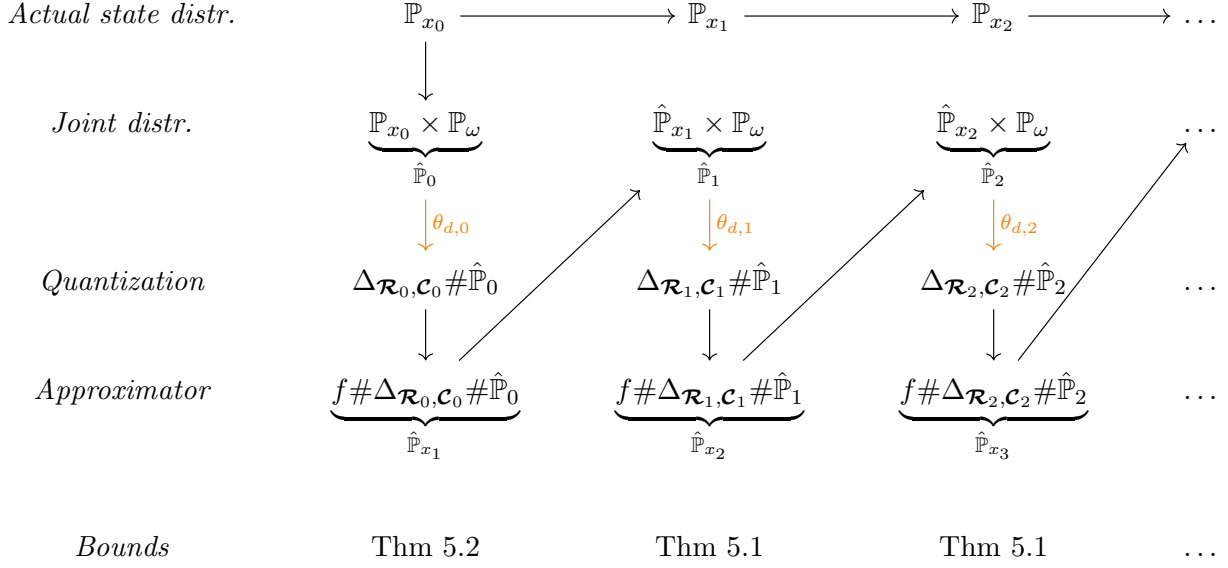
%In Theorem \ref{thm:how-to-propagate-stochastic-systems}, we adopt the same notation as in Figure \ref{fig:approximation-scheme-dynamical-systems} and as described in the previous paragraph. For each set of locations $\bsC^t$, we denote $N^t \coloneqq |\bsC^t|$. Further, we assume that the coefficients $(\alpha^0_k, \beta^0_k)$ are chosen according to the norm linearization constraint Eqn \eqref{} in Theorem \ref{th:bound-zero-ball}, while $(\alpha^t_k, \beta^t_k)$ satisfy Eqn \eqref{} in Theorem \ref{th:bound}.
\begin{theorem}[Approximation error dynamics]\label{thm:how-to-propagate-stochastic-systems}
Given $\epsilon>0$, let $\bsR_t$ be $(\sX \times \sW)$-partitions and $\bsC_t \subset \sX \times \sW$ sets of locations such that $\theta_{d,t} = \bigg( \sum_{k=1}^{N_t} \int_{\sR_{t, k}} \norm{x-c_{t, k}}^\rho d\probHat_t(x) \bigg)^\frac{1}{\rho} \leq \epsilon$ for every $t \geq 0 $. Consider the following iterative process describing the approximation error evolution for $t\in \mathbb{N}_{>0}$:
\begin{align*}
    & \theta_1 = {\left( \sum_{k=1}^{N_0} \alpha_{0, k} \int_{\sR_{0, k}} \norm{x-c_{0, k}}^\rho d\probHat_0(x) + \sum_{k=1}^{N_0} \probHat(\sR_{0, k})\beta_{0,k} \right)^\frac{1}{\rho} }, \\
    & \theta_{t+1} = \left(\alpha_{\text{max},t}( \theta_t + \epsilon )^\rho + \sum_{k=1}^{N_t} \probHat(\sR_{t, k}) \beta_{t, k} \right)^\frac{1}{\rho}.
\end{align*}
Then, for any $t>0$, the following holds:

i) $\wasserstein_\rho(\prob_{x_{t}}, \probHat_{x_{t}}) \leq \theta_t \,$.

ii) If the dynamics $f$ in \eqref{eq:system-def} is Lipschitz continuous in $(x, \omega)$ with constant $\sL_f <1$, then
$$ \lim_{t \to \infty} \wasserstein_\rho(\prob_{x_{t}}, \probHat_{x_{t}}) \leq   \frac{\sL_f}{1- \sL_f }\epsilon \,.$$
\end{theorem}
The proof of Theorem~\ref{thm:how-to-propagate-stochastic-systems} is reported in Appendix~\ref{proof:thm:how-to-propagate-stochastic-systems} and follows from a combination of Theorem \ref{th:bound} and \ref{th:bound-zero-ball} with the Banach Fixed Point Theorem \cite{goebel1990topics}.
Theorem \ref{thm:how-to-propagate-stochastic-systems} has many consequences. First of all, the bound does not necessarily grow with time; it is possible that $\theta_{t+1}<\theta_t$ if the dynamics contracts sufficiently. This is a fundamental advantage with respect to existing approaches for the same problem, whose bounds tend to grow linearly with time \cite{figueiredo2024uncertainty}. 
Furthermore, Theorem \ref{thm:how-to-propagate-stochastic-systems} guarantees that if $f$ is contracting w.r.t. $(x,\omega)$, i.e., $\lipschitz_f<1$, then the approximation error will reach a fixed point. 
Note also that the bound for the fixed point of the error reported in case ii) in Theorem \ref{thm:how-to-propagate-stochastic-systems} is stated using 
the linearization coefficients from Remark \ref{rem:linearSyst}, i.e., \((\alpha_k,\beta_k)=(\lipschitz_f, 0)\) for all \(k\). 
Consequently, in practice, the approach in Figure \ref{fig:approximation-scheme-dynamical-systems} may yield a smaller bound. 
Notably, as empirically shown in Section \ref{section:experimental-results}, our approach can lead to a fixed point for $\theta_d$ even when $\sL_f>1$ if $f$ is bounded.
% This will be shown empirically in Section \ref{section:experimental-results}, where we show how the approach in Figure \ref{fig:approximation-scheme-dynamical-systems} can lead to a fixed point for $\theta_d$ even when $\sL_f>1$.

%More importantly, for contractive $f$, claim ii) shows that the error between the true state distribution $\prob_{x_t}$ and our approximator $\probHat_{x_t}$ remains bounded in the long-term, and can actually be made arbitrarily small as we have control on $\epsilon$.

\begin{remark}[Separable dynamics]
    For a process with separable dynamics as $f(x, \omega) = g(x) + s(\omega)$, where $g$ and $s$ are given piecewise Lipschitz continuous functions, we observe:
    \begin{align}  
        \maybeamp\wasserstein_\rho(\prob_{x_{t+1}}, \probHat_{x_{t+1}}) \nonumber \maybenewline &= \wasserstein_\rho(g \#\prob_{x_t} * s\#\prob_\omega, g \# \Delta_{\bsR, \bsC} \# \probHat_{x_t} * s\#\Delta_{\bsR_\omega, \bsC_\omega}\#\probHat_\omega)\maybenonumber \\
        &\leq\wasserstein_\rho(g \#\prob_{x_t}, g \# \Delta_{\bsR, \bsC} \# \probHat_{x_t}) \, + \nonumber \\
        & \hspace{3.5cm}\wasserstein_\rho(s \#\prob_{\omega}, s \# \Delta_{\bsR_\omega, \bsC_\omega} \# \probHat_{\omega}), \label{eq:28}
    \end{align}
    where $*$ is the convolution operator, $\bsR, \bsC$ defined in $\sX$-space, and $\bsR_\omega, \bsC_\omega$ in $\sW$. When $\prob_\omega$ is known, the right term in \eqref{eq:28} is constant for all $t$ and only needs to be computed only once.
    % (thus, only needs to be computed once).
\end{remark}

\begin{remark}[Ambiguous noise]
    Although we consider both $\prob_{x_0}$ and $\prob_\omega$ are known in this section, the framework can be easily extended to case where one has uncertain $\prob_{x_0} \in \mathbb{B}_{\theta_0}(\Tilde{\prob})$ and $\prob_{\omega} \in \mathbb{B}_{\theta_\omega}(\Tilde{\mathbb{T}})$, where $\theta_0, \theta_\omega>0$, $\Tilde{\prob}\in \sP_\rho(\sX)$, and $\Tilde{\mathbb{T}} \in \sP_\rho(\sW)$ are given. In this case, we note that $\prob_{x_0} \times \prob_\omega \in \mathbb{B}_{\theta_0+\theta_\omega}(\Tilde{\prob} \times \Tilde{\mathbb{T}})$ and then we use Theorem \ref{th:bound} to bound the first time-step $\wasserstein_\rho(f\#(\prob_{x_0} \times \prob_\omega), f\#\Delta_{\bsR, \bsC}\#(\prob_{x_0} \times \prob_\omega))$.
    % (while Theorem \ref{th:bound-zero-ball} can be employed when they are known exactly).
\end{remark}

\section{Experimental results}\label{section:experimental-results}
In this Section, we empirically evaluate the performance of our $\rho$-Wasserstein uncertainty propagation framework on various benchmarks taken from the literature\footnote{Our code is available at \url{https://github.com/sjladams/DUQviaWasserstein}}. We consider the following piecewise Lipschitz continuous functions $f$: a \textit{Bounded Linear} $f$ adapted from \cite{santoyo2021barrier} with state space dimension $d$ ranging from $1$ to $4$, an instance of the \textit{Quadruple-Tank} from \cite{johansson2000quadruple}, the \textit{Mountain Car} dynamics \cite{singh1996reinforcement}, and the \textit{Dubins Car} \cite{balkcom2018dubins}. Additionally, we consider the \textit{Sigmoid} function introduced in Example \ref{exampl:ImportanceOdAlphaBeta}, and a 10-dimensional \textit{Neural Network layer}. In Section \ref{subsection:stochastic-system-experiments-results}, we consider stochastic dynamical systems variants of the Mountain Car \cite{singh1996reinforcement}, and Quadruple-Tank \cite{johansson2000quadruple} with additive Gaussian noise, and of a 3D-NN dynamics affected by non-Gaussian process noise. Additional details about the functions, dynamical systems, and probability distributions are available in the Appendix~\ref{appendix:functions}.

In what follows, first, in Sections \ref{subsection:applying-loc-selection-algo} and \ref{subsection:convergence} we investigate the impact of the placement and the number of quantization locations \(\bsC\) on the error bounds, respectively. 
Then, in Section \ref{subsection:experiments-impact-theta}, we analyze the effect of the linearization coefficients in Theorems \ref{th:bound} and \ref{th:bound-zero-ball} in case of non-linear functions \(f\) for different radii of uncertainty \(\theta\). 
Lastly, in Section \ref{subsection:stochastic-system-experiments-results}, we apply the approximation scheme presented in Section \ref{section:dynamical-systems} to stochastic dynamical systems.   For all the experiments, we fix $\rho=2$.  All the experiments were conducted on an Intel Core i7-1365U CPU with 16GB of RAM using a single-core implementation.

\begin{table}[hb]
\centering
\ifthenelse{\boolean{doublecolumn}}{}{\small}
\begin{tabular}{llcccc}
                     &                    & \multicolumn{4}{c}{$|\bsC|$}                          \\ \cline{3-6} 
Dim. $d$ & Algorithm & 5 & 10 & 100 & 1000 \\ \hline
\multirow{2}{*}{1}   & Optimized grid        & 0.5085     & 0.2731      & 0.0280       & 0.0087        \\
                     & Uniform grid       & 0.5420     & 0.3363      & 0.0487       & 0.0169        \\ \hline
\multirow{2}{*}{2}   & Optimized grid        & 0.7867     & 0.1935      & 0.0723       & 0.0248        \\
                     & Uniform grid       & 0.7867     & 0.3826      & 0.1566       & 0.0539        \\ \hline
\multirow{2}{*}{3}   & Optimized grid        & 0.7940     & 0.1982      & 0.0818       & 0.0410        \\
                     & Uniform grid       & 0.7940     & 0.5428      & 0.3532       & 0.1801        \\ \hline
\multirow{2}{*}{4}   & Optimized grid        & 1.8681     & 0.8043      & 0.4078       & 0.2111        \\
                     & Uniform grid       & 1.8681     & 1.8465      & 0.7935       & 0.6161        \\ \hline
\end{tabular}
\caption{Comparison of error bounds from Theorem \ref{th:bound-zero-ball} for $\bsR, \bsC$ obtained 
as described in Remark~\ref{remark:algorithm-quantization-gaussian-cas} (\emph{Optimized grid}) and the uniform partition (\emph{Uniform grid}) for the Bounded Linear benchmark defined in the Appendix~\ref{appendix:functions}.}
\label{table:compare-alg1-to-uniform-grid}
\end{table}

\begin{figure}[ht]
    \centering
    \ifthenelse{\boolean{doublecolumn}}{
        \includegraphics[width=1.0\linewidth]{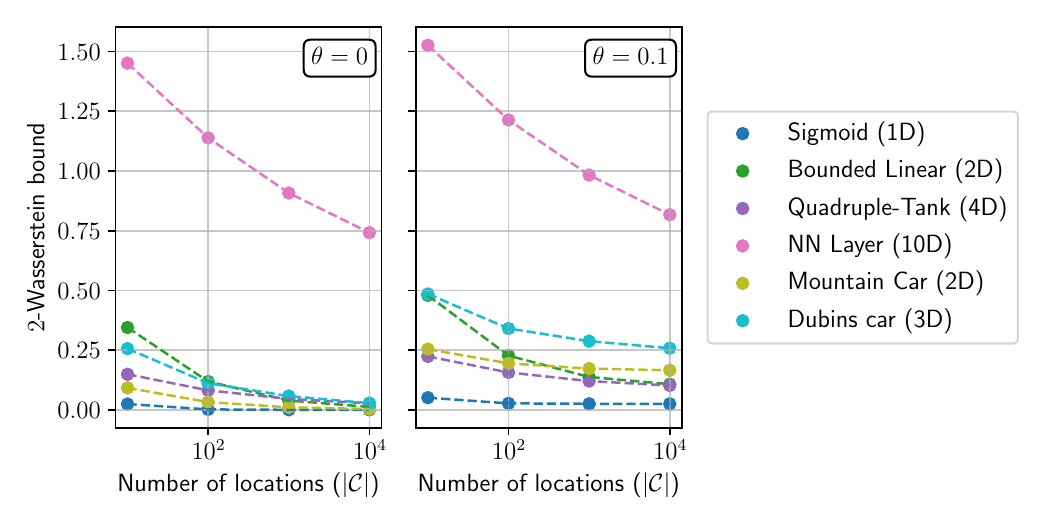}
    }{
        \includegraphics[width=0.8\linewidth]{Figures/Experiments/increase_num_locs_analysis.pdf}
    }
    \caption{Upper bounds on $\sup_{\probQ \in \mathbb{B}_\theta(\prob)} \wasserstein_2(f\#\probQ, f\#\Delta_{\bsR, \bsC}\#\prob)$ 
    for various benchmarks computed using Theorem \ref{th:bound-zero-ball}  for $\theta = 0$ and Theorem \ref{th:bound} for $\theta = 0.1$. }
    \label{fig:effect-adding-locs}
\end{figure}

\subsection{Improving on uniform grids of quantization locations}\label{subsection:applying-loc-selection-algo}
In this Section, we analyze the effect of optimizing the locations \(\bsC\) used for the quantization operator \(\signature_{\bsR, \bsC}\) using the approach in Remark~\ref{remark:algorithm-quantization-gaussian-cas} compared to taking a uniform grid. 
In particular, in Table \ref{table:compare-alg1-to-uniform-grid}, for a bounded linear $f: \realNum^d \to \realNum^d$ defined in Appendix \ref{subsection:further-details} for each $d \in \{1,2,3,4\}$, for different quantization sizes $|\bsC|$, we compare
the error bound from Theorem~\ref{th:bound-zero-ball}, 
obtained using the the procedure described in Remark~\ref{remark:algorithm-quantization-gaussian-cas}, 
with that obtained from a uniform partition of a subset $\Tilde{\sX} \subset \sX$ containing most of the probability mass of $\prob$\footnote{This uniform partition is defined as follows. We first get $\bsC$ from Remark~\ref{remark:algorithm-quantization-gaussian-cas}. We then move the locations $c_k \in \bsC$ such that they are equally spaced in all axes (also forming a grid), obtaining $\bsC_{\text{unif}}$. Finally, we compute $\bsR_{\text{unif}}$ as the Voronoi partition w.r.t. $\bsC_{\text{unif}}$.}. From Table \ref{table:compare-alg1-to-uniform-grid}, we observe that as the dimensionality of the problem increases, the restrictiveness of placing locations in an equidistant fashion also augments. In fact, note that while for $d=1$ the error bound in Theorem \ref{th:bound-zero-ball} is similar regardless of the heuristics used to place the locations, for $d=4$ the selection performed by employing Remark \ref{remark:algorithm-quantization-gaussian-cas} leads to bound 2-3 times smaller than the uniform partition approach.

\subsection{Error bound convergence}\label{subsection:convergence}
In the previous Section, we focused on the placement of the quantization operator's locations. Here, we analyze how the $2$-Wasserstein bounds decrease as the number of optimized locations for Theorems \ref{th:bound} and \ref{th:bound-zero-ball} grows. More precisely, given a distribution $\prob$ and a ambiguity set $\mathbb{B}_\theta(\prob)$ of radius $\theta=0$ or $\theta=0.1$, we report the bound of $\sup_{\probQ \in \mathbb{B}_\theta(\prob)} \wasserstein_2(f\#\probQ, f\#\Delta_{\bsR, \bsC}\#\prob)$ for different quantization sizes $|\bsC|$. 

From Figure~\ref{fig:effect-adding-locs}, we observe that for all benchmarks increasing the number of locations in the quantization leads to a decreasing bound. This is expected due to the reduction of $\theta_d$ guaranteed by the discussion in Section \ref{section:constructing-converging-quant-operator}. In the case where $\theta=0$, as there is no uncertainty around $\prob$, the bounds converge to zero. In contrast, with $\theta_d=0.1$, the bounds do not converge to zero, but to different values for each system. Both observations empirically confirm Theorem \ref{prop:convergence-algorithm-error-approx}. It is also important to note that the error bounds are impacted both by the geometry of the probability space of \(\prob\) as well as the system dynamics \(f\). For instance, for the Dubins car, the upper bound on $\sup_{\probQ \in \mathbb{B}_\theta(\prob)} \wasserstein_2(f\#\probQ, f\#\Delta_{\bsR, \bsC}\#\prob)$ is consistently larger than that of the Quadruple-Tank, even though the Quadruple-Tank is higher dimensional. This can be explained because the Dubins car is not a stable system, and, consequently, the resulting uncertainty in terms of $2$-Wasserstein distance is amplified.

\begin{figure}[htb]
    \centering
    \ifthenelse{\boolean{doublecolumn}}{
        \includegraphics[width=1.0\linewidth]{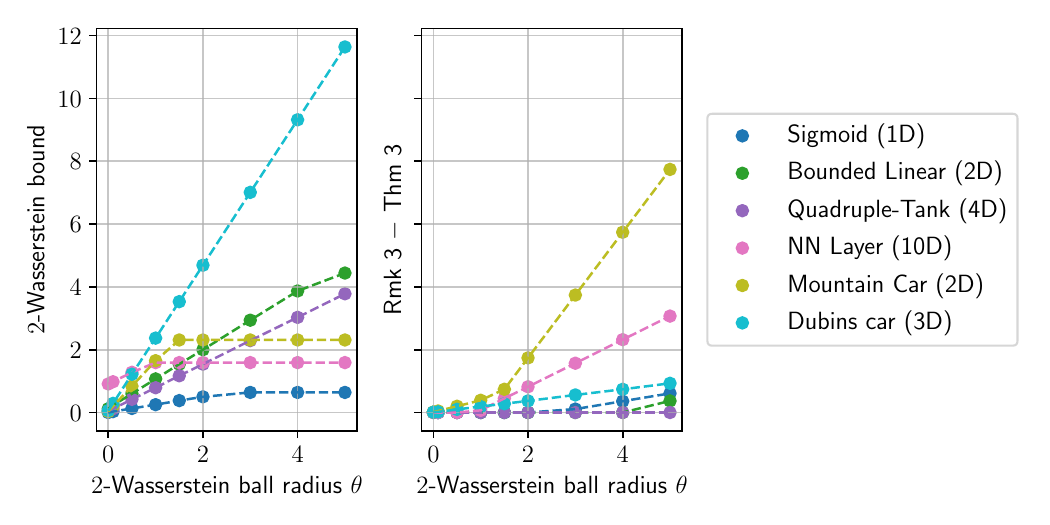}
    }{
        \includegraphics[width=0.8\linewidth]{Figures/Experiments/wass_ball_radius_analysis.pdf}
    }
    \caption{Analysis of the upper bounds on $\sup_{\probQ \in \mathbb{B}_\theta(\prob)} \wasserstein_2(f\#\probQ, f\#\Delta_{\bsR, \bsC}\#\prob)$ computed using Theorem \ref{th:bound-zero-ball} for $\theta = 0$ and Theorem \ref{th:bound} for $\theta >0$. In the left plot are the absolute bounds using the linearization coefficients from Section \ref{subsection:approx-algo}; on the right, the absolute difference between the bounds using the more conservative global Lipschitz coefficients.}
    \label{fig:propagate-ball-theta-sizes}
\end{figure}

\begin{table*}[t]
    \centering
    \ifthenelse{\boolean{doublecolumn}}{}{\small}
    \begin{tabular}{lccccccccc}
    \textbf{}          & \multicolumn{3}{c}{\textbf{NN Layer (3D)}}                                & \multicolumn{3}{c}{\textbf{Mountain Car}}                                          & \multicolumn{3}{c}{\textbf{Quadruple-Tank}}     \\ \cline{2-10} 
    $t$ & \textbf{Emp.} & \textbf{Rmk 1} & \multicolumn{1}{c|}{\textbf{Thm 4}} & \textbf{Emp.} & \textbf{Rmk 1}               & \multicolumn{1}{c|}{\textbf{Thm 4}} & \textbf{Emp.} & \textbf{Rmk 1} & \textbf{Thm 4} \\ \hline
    1                  & 0.0116             & 0.2020         & \multicolumn{1}{c|}{0.1214}         & 0.0256             & 0.0627                       & \multicolumn{1}{c|}{0.0547}         & 0.0821             & 0.1517         & 0.1517         \\
    2                  & 0.0090             & 0.2436         & \multicolumn{1}{c|}{0.1358}         & 0.0302             & 0.2364                       & \multicolumn{1}{c|}{0.1826}         & 0.0790             & 0.2748         & 0.2748         \\
    3                  & 0.0102             & 0.2732         & \multicolumn{1}{c|}{0.1464}         & 0.0498             & 0.6183                       & \multicolumn{1}{c|}{0.4178}         & 0.0757             & 0.3670         & 0.3670         \\
    4                  & 0.0102             & 0.2941         & \multicolumn{1}{c|}{0.1522}         & 0.0371             & 1.3944                       & \multicolumn{1}{c|}{0.8088}         & 0.0680             & 0.4308         & 0.4308         \\
    5                  & 0.0104             & 0.3097         & \multicolumn{1}{c|}{0.1555}         & 0.0413             & 2.9423                       & \multicolumn{1}{c|}{1.4388}         & 0.0643             & 0.4751         & 0.4751         \\
    6                  & 0.0105             & 0.3213         & \multicolumn{1}{c|}{0.1574}         & 0.0433             & 6.0399                       & \multicolumn{1}{c|}{2.4609}         & 0.0621             & 0.5031         & 0.5031         \\
    7                  & 0.0106             & 0.3301         & \multicolumn{1}{c|}{0.1586}         & 0.0407             & 12.2351                      & \multicolumn{1}{c|}{2.9560}         & 0.0618             & 0.5185         & 0.5185         \\
    8                  & 0.0102             & 0.3366         & \multicolumn{1}{c|}{0.1593}         & 0.0507             & 24.6256                      & \multicolumn{1}{c|}{2.9748}         & 0.0659             & 0.5260         & 0.5260         \\
    9                  & 0.0105             & 0.3414         & \multicolumn{1}{c|}{0.1595}         & 0.0505             & 49.4063                      & \multicolumn{1}{c|}{2.9910}         & 0.0793             & 0.5242         & 0.5242         \\
    10                 & 0.0099             & 0.3451         & \multicolumn{1}{c|}{0.1598}         & 0.0758             & 98.9695                      & \multicolumn{1}{c|}{3.0035}         & 0.0769             & 0.5174         & 0.5174         \\ \hline
    50                 & 0.0100             & 0.3562         & \multicolumn{1}{c|}{0.1601}         & 0.0676             & $1.1 \times 10^{14}$ & \multicolumn{1}{c|}{3.1819}         & 0.0767             & 0.4794         & 0.4794        
    \end{tabular}
    \caption{Formal bounds on $\wasserstein_2(\prob_{x_t}, \probHat_{x_t})$ from Theorem \ref{th:bound} using the linearization coefficients described in Section \ref{subsection:approx-algo}, as shown in column \textit{Thm 4}, or employing coefficients \((\lipschitz_f,0)\), as in column \textit{Rmk 1}. Column \textit{Emp.} presents a Monte Carlo approximation of $\wasserstein_\rho(\prob_{x_t}, \probHat_{x_t})$, calculated using $5 \times 10^5$ samples.
    }
    \label{table:results-bounds-system-10-time-steps}
\end{table*}

\subsection{Analysis of ambiguity set propagation using global and local linearization}\label{subsection:experiments-impact-theta}
We continue our analysis by investigating the impact of the linearization coefficients on our 2-Wasserstein bounds for different uncertainty radii $\theta$. 
Specifically, we compare the bounds constructed using the trivial linearization coefficients \((\lipschitz_f,0)\), with those derived from the coefficients described in Section \ref{subsection:approx-algo}, as per Theorem \ref{th:bound} and \ref{th:bound-zero-ball}\footnote{More specifically, we report $\sL_f (\theta + \theta_d) - \bigg( \alpha_{\max} (\theta+\theta_d)^\rho + \sum_{k=1}^{N} \evpi{k}\beta_k \bigg)^\frac{1}{\rho}$ for $\theta>0$, and $\sL_f (\theta + \theta_d) - \left( \sum_{k=1}^N \alpha_k \int_{\sR_k} \norm{x-c_k}^\rho d\prob(x) + \sum_{k=1}^N \evpi{k} \beta_k  \right)^\frac{1}{\rho}$ for $\theta=0$.}. 
We set $|\bsC| = 10^2$ for functions with dimension of at most three, and use $|\bsC| = 10^3$ otherwise. The $\realNum^d$-partitions $\bsR$ and locations $\bsC \subset \realNum^d$ are selected as outlined in Remark \ref{remark:algorithm-quantization-gaussian-cas}.

The left plot of Figure \ref{fig:propagate-ball-theta-sizes} shows that for the optimized coefficients in case of bounded functions (NN Layer, Mountain Car and Bounded Linear), the bounds saturate from a certain $\theta$ onwards. This saturation occurs because, for large $\theta$, we select $(\alpha_k, \beta_k)=(0,\sup_{x\in \sX}\norm{f(x)-f(c_k)}^\rho)$ for most regions, as explained in Section \ref{subsection:approx-algo}. Consequently, the error bound from Theorem \ref{th:bound} becomes independent of \(\theta\).
Furthermore, it is important to note that in many cases, the error bounds are smaller than $\theta$, which indicates a contraction of the ambiguity set. An exception is the Dubins car example, where instability in the system dynamics causes the ambiguity set to expand.

The right plot of Figure \ref{fig:propagate-ball-theta-sizes} confirms that, as discussed in Remark \ref{rem:linearSyst}, for nonlinear systems, the bounds constructed using the optimized coefficients are consistently and substantially tighter that the bound resulting from using the global Lipschitz coefficients. Note that for linear systems (Quadruple-Tank), the two coefficients are equivalent and lead to the same linearizations in \eqref{eq:norm-linearization} and \eqref{eq:norm-inequality-zero-budget}.

\subsection{Uncertainty Propagation in Stochastic Dynamical Systems}\label{subsection:stochastic-system-experiments-results}
In this Section, we apply the discrete approximation scheme presented in Section \ref{section:dynamical-systems} and illustrated in Figure \ref{fig:approximation-scheme-dynamical-systems} to three stochastic dynamical systems.
We analyze both an empirical estimation of, and our formal bounds on, $\wasserstein_2(\prob_{x_t}, \probHat_{x_t})$, where $\prob_{x_t}$ represents the true unknown state distribution at time $t$ and $\probHat_{x_t}$ is our discrete approximator.
In Table~\ref{table:results-bounds-system-10-time-steps}, we observe that the empirical $\rho$-Wasserstein distance remains low over longer time horizons, demonstrating the effectiveness of the approximation in practical scenarios. 
For the contracting NN Layer and Quadruple-Tank dynamics, the Monte Carlo estimates of the approximation error converge to fixed points, supporting Theorem \ref{thm:how-to-propagate-stochastic-systems}. 
For the non-contracting ($\lipschitz_f>1$) but bounded Mountain Car dynamics, the bounds from Theorem \ref{th:bound} obtained using coefficients $(\lipschitz_f,0)$ quickly explode. In contrast, using Theorem \ref{th:bound} results in bounds that converge to a fixed point due to the boundedness of the dynamics. 
From Figure \ref{fig:mountain-car-multistep-100}, we can visually confirm that our discrete approximators (right column) closely match an empirical estimate of true distributions (left column). 
We highlight that the discrete approximator is able to capture the fact that the state distribution becomes bimodal at $t=10$. Such characteristics are challenging to identify using techniques like moment matching \cite{deisenroth2011pilco}, for instance, which only focus on approximating, commonly with no guarantees, the first few moments of the distribution.

\begin{figure}[h]
  \centering
  \vspace{-1mm}
  \ifthenelse{\boolean{doublecolumn}}{
  \includegraphics[width=0.85\linewidth]{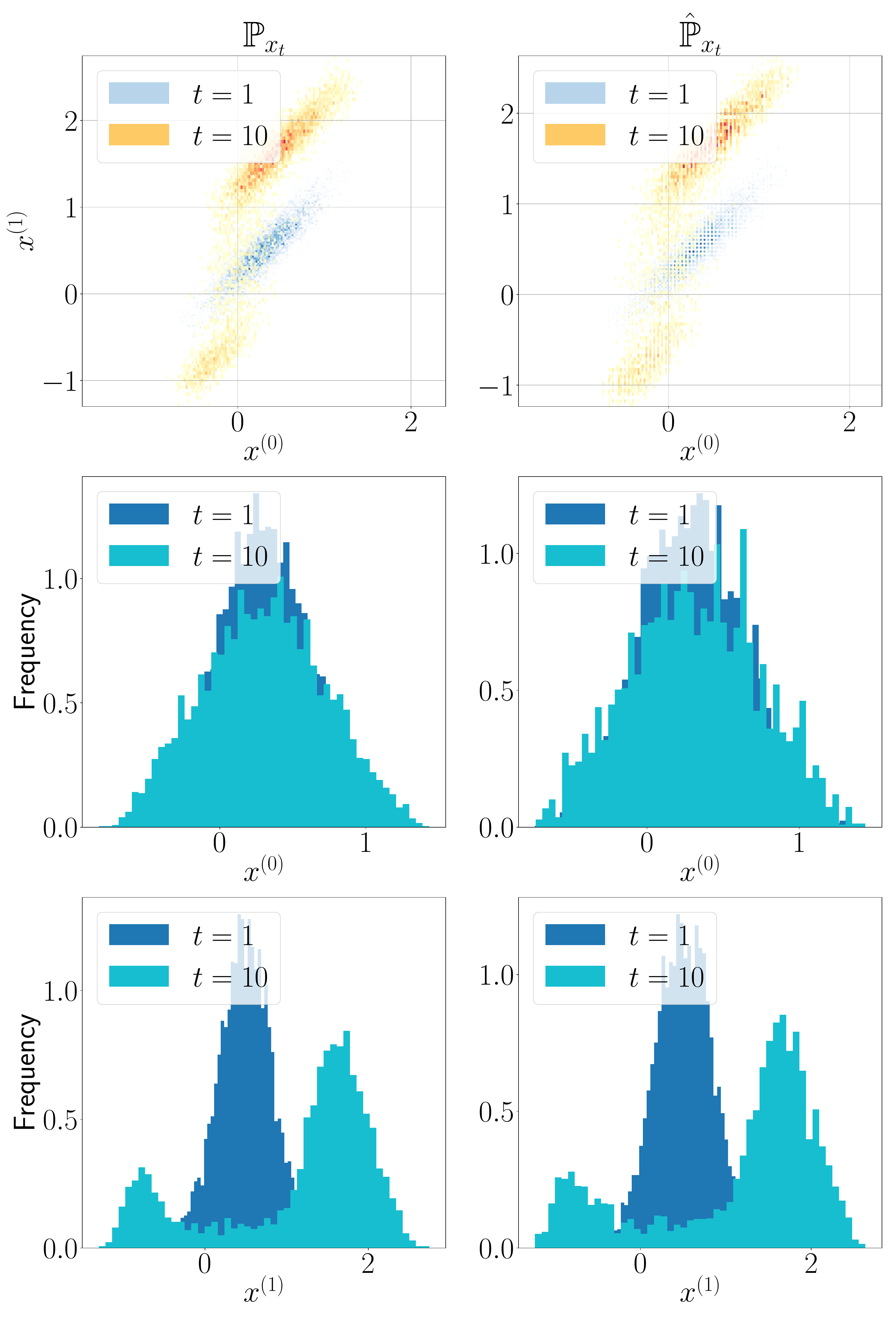}
  }{
  \includegraphics[width=0.55\linewidth]{Figures/Experiments/multistep_mountain-car.pdf}
  }
  \vspace{-3mm}
  \caption{
  Monte Carlo simulation of the true state distribution (left plots) - with $5 \times 10^3$ samples - and our discrete approximation from Section \ref{section:dynamical-systems} (right plots) - with $|\bsC|=100$ - for the Mountain Car system from $t=1$ to $t=10$. The upper plots display the joint distribution of the first two state dimensions, $x^{(1)}_t$ and $x^{(2)}_t$, for all time steps, and the lower plots illustrate the initial and final marginal distributions. 
  }
  \label{fig:mountain-car-multistep-100}
\end{figure}

\section{Conclusion and future direction}\label{section:conclusion}
We introduced a novel framework to approximate the push-forward measure of uncertain distributions with discrete distributions with formal quantification of the resulting uncertainty in terms of $\rho$-Wasserstein distance, allowing for a tractable propagation of $\rho$-Wasserstein ambiguity sets. 
We see at least three interesting future research directions. 
% In the following, we indicate some potential future research directions. 
First, the development of efficient ways to compute the norm approximations in \eqref{eq:norm-linearization}. 
% First, our approach is critically dependent on tight local norm linearizations in \eqref{eq:norm-linearization}, which are usually hard to compute.  
% In this work, we employed linear bound propagation techniques, but we conjecture that Taylor expansion-based techniques could reduce conservativeness. 
Further, in the context of multi-step propagation of ambiguity sets, such as for dynamical systems, it may be of interest to directly rely on properties of the compositions of $f \circ ... \circ f$, instead of the sequential application of our framework, as we propose in Section \ref{section:dynamical-systems}. Lastly, we indicate that this framework could be directly applied as the prediction mechanism for distributionally-robust non-linear Model Predictive Control.

\appendix
\section{Proofs}\label{section:proofs}
In this section, we present the proofs for all the results discussed in the paper's main text.

\subsection{Proof of Proposition \ref{prop:compute-theta-d}}\label{proof:prop:compute-theta-d}
Before proving Proposition \ref{prop:compute-theta-d}, we prove an auxiliary Lemma.
\begin{lemma}\label{lemma:coupling-signature}
    For $\sX \subseteq \realNum^d$, let $\prob \in \sP_{\rho}(\sX)$. Further, let $\bsR$ be a $\sX$-partition and $\bsC$ a set of locations. Then, for $\gamma^*\in\probMeas(\sX\times\sX)$ defined as
    \begin{equation}\label{eq:general-coupling}
        d\gamma^*(x_1, x_2) \coloneqq \sum_{i=1}^N\indicator_{\sR_i}(\vx_1) d\prob(x_1) d\delta_{c_i}(x_2)
    \end{equation}
    it holds that:
    
    i) $\gamma^*$ is a valid coupling between $\prob$ and $\Delta_{\bsR, \bsC}\#\prob$, i.e. $\gamma^* \in \Gamma(\prob, \Delta_{\bsR, \bsC}\#\prob)$,
    
    ii) if $\Bar{\bsR}$ is the Voronoi partition w.r.t. $\bsC$ then  
    \[
        \gamma^* = \arginf{\gamma \in \Gamma(\prob,\Delta_{\Bar{\bsR},\sC}\#\prob)} \int_{\sX \times \sX} \norm{x-y}^{\rho} d\gamma(x, y).
    \]
\end{lemma}
\begin{proof}
  We start  by proving that $\gamma^* \in \Gamma(\prob, \Delta_{\bsR, \bsC}\#\prob)$. For any $A, B \in \sB(\sX)$, we have:
    \begin{align*}
        \gamma^*(A, B) &= \int_{A} \int_{B} d\gamma^*(\vx_1, \vx_2) \\ 
        &= \int_{A} \int_{B} \sum_{i=1}^N\indicator_{\sR_i}(\vx_1)d\prob(x_1)d\delta_{\vc_i}(\vx_2) \\ &=\sum_{i=1}^N\prob(A\cap\sR_i)\indicator_{B}(\vc_i),
    \end{align*}
    which is a value in $[0, 1]$ since $\bsR$ is a partition of $\sX$. Further, note that:
    \begin{alignat*}{2}
        \gamma(A,\sX) &=\sum_{i=1}^N\prob(A\cap\sR_i)\indicator_{\sX}(\vc_i) =\sum_{i=1}^N\prob(A\cap\sR_i)=\prob(A)
    \end{alignat*}
    and
    \begin{alignat*}{2}
        \gamma(\sX,B) \maybeamp=\sum_{i=1}^N\prob(\sX\cap\sR_i)\indicator_{B}(\vc_i) =\sum_{i=1}^N\prob(\sR_i)\indicator_{B}(\vc_i) \maybenewline\maybeamp=\big( \Delta_{\bsR, \bsC}\#\prob \big)(B)
    \end{alignat*}
    and, consequently, $\gamma(\sX, \sX) = 1$.
    Thus,  $\gamma \in \Gamma(\prob, \Delta_{\bsR, \bsC}\#\prob)$. This proves item i).
    To prove item ii), it suffices to note that if $x \in \sR_i$, then by the definition of the Voronoi partition w.r.t. $\bsC$, the cost of transporting $d\prob(x)$ to $c_i$ is smaller than any other $c_j, j \neq i$ since $\norm{x-c_i} \leq \norm{x-c_j}$.
\ifthenelse{\boolean{doublecolumn}}{\qed}{}
\end{proof}
We are now ready to prove Proposition \ref{prop:compute-theta-d}.
Let $\gamma^*$ be defined as in \eqref{eq:general-coupling}. Using item i) from Lemma \ref{lemma:coupling-signature}, we have:
\begin{align}
    & \wasserstein_\rho(\prob, \Delta_{\bsR, \bsC}\#\prob)^\rho   \leq \int_{\sX \times \sX} \norm{x-y}^{\rho} d\gamma^*(x, y) \label{eq:ineq-prove-prop12} \\ 
    & = \int_{\sX \times \sX} \norm{x-y}^{\rho} \sum_{k=1}^N\indicator_{\sR_k}(\vx) d\prob(x) d\delta_{c_k}(y) \nonumber \\ 
    & =\sum_{k=1}^N \int_{\sR_k} \norm{x-c_k}^\rho d\prob(x) \nonumber 
\end{align}
If $\Bar{\bsR}$ is the Voronoi partition w.r.t. $\bsC$, by applying item ii) from Lemma \ref{lemma:coupling-signature}, the inequality in \eqref{eq:ineq-prove-prop12} becomes an equality.
\qed

\subsection{Proof of Theorem \ref{th:bound}}\label{proof:th:bound}
Before proving Theorem \ref{th:bound}, we show that $\sup_{\probQ \in \mathbb{B}_\theta(\prob)} \wasserstein_\rho(f\#\probQ, f\#\Delta_{\bsR, \bsC}\#\prob)$ can be upper bounded by a one-dimensional minimization program, using duality techniques from the DRO literature \citep{gao2023distributionally, mohajerin2018data, yue2022linear}.
\begin{proposition}\label{proposition:lagrangian-duality}
    For $\sX \subseteq \realNum^d$, let $\prob \in \sP_\rho (\sX)$, $\bsR$ be a $\sX$-partition, and $\bsC$ be a set of locations. Further, denote $\theta_d \coloneqq \bigg( \sum_{k=1}^N \int_{\sR_k} \norm{x-c_k}^\rho d\prob(x) \bigg)^\frac{1}{\rho}$, and call $\evpi{i} \coloneqq \prob(\sR_i)$ for every $\sR_i \in \bsR$. Then,
    \begin{align}
        &\sup_{\probQ \in \mathbb{B}_\theta(\prob)} \wasserstein_\rho(f\#\probQ, f\#\Delta_{\bsR, \bsC}\#\prob) \nonumber\\
        &\qquad\leq \bigg( \inf_{\lambda \geq 0} \bigg\{ \lambda (\theta+\theta_d)^\rho \maybenonumber\maybenewline\maybeamp\maybeqquad + \sum_{j=1}^N \evpi{j} \sup_{\xi \in \sX} \big( \norm{f(\xi) - f(\vc_j)}^\rho - \lambda \norm{\xi - \vc_j}^\rho \big) \bigg\}  \bigg)^\frac{1}{\rho} \maybenonumber
    \end{align}
\end{proposition}
\begin{proof}
    We define $S_\theta(\prob)$ as a subspace of $\sP(\sX \times \sX)$ containing all the couplings for which one of the marginals is $\prob$ and the other implied marginal is at most $\theta$ far in $\rho$-Wasserstein distance from $\prob$, i.e.
    \begin{align*}
        S_\theta(\prob) \coloneqq\maybeamp\Big\{ \gamma \in \sP(\sX \times \sX) : \maybenonumber\maybenewline\maybeamp\int_{\sX \times \sX} \norm{x_1 - x_2}^\rho d\gamma(x_1, x_2) \leq \theta^\rho, \text{proj}_2 \# \gamma = \prob \Big\},
    \end{align*}
    where $\text{proj}_2 \# \gamma$ returns the marginal distribution of $\gamma$ in the second component, i.e. $\text{proj}_2 \# \gamma \coloneqq \int_{\sX} \gamma(d x_1, .)$. We then have:
    \begin{align}   
        &\bigg( \sup_{\probQ \in \mathbb{B}_{\theta}(\prob)}\wasserstein_\rho(f\#\probQ,f \#\signature_{\bsR, \bsC}\#\prob) \bigg)^\rho \nonumber \\
        & \qquad \text{(By monotonicity of $x^{\rho}$ for $x\geq 0$)} \nonumber\\
        &= \sup_{\probQ \in \mathbb{B}_{\theta}(\prob)}\wasserstein_\rho(f\#\probQ,f \#\signature_{\bsR, \bsC}\#\prob)^\rho \nonumber \\
        & \qquad \text{ 
        (Using Eqn~\eqref{eq:wass-for-pushforward-equivalence})
        % (By Proposition \ref{prop:wass-for-pushforward-equivalence})
        } \nonumber\\
        &= \sup_{\probQ \in \mathbb{B}_{\theta}(\prob)} \inf_{\gamma \in \Gamma(\probQ, \Delta_{\bsR, \bsC}\#\prob)} \int_{\sX \times \sX} \norm{f(x_1) - f(x_2)}^\rho d\gamma(x_1, x_2) \nonumber \\
        & \, \text{(As $\mathbb{B}_{\theta}(\prob) \subseteq \mathbb{B}_{\theta+\theta_d}(\Delta_{\bsR, \bsC} \# \prob)$ for $\wasserstein_\rho(\prob, \Delta_{\bsR, \bsC} \# \prob) \leq \theta_d$) } \nonumber\\
        &\leq \sup_{\probQ \in \mathbb{B}_{\theta+\theta_d}(\Delta_{\bsR, \bsC}\#\prob)} \inf_{\gamma \in \Gamma(\probQ, \Delta_{\bsR, \bsC}\#\prob)} \maybenonumber\maybenewline
        \maybeamp\maybeqquad\maybeqquad\maybeqquad\maybeqquad\maybeqquad\int_{\sX \times \sX} \norm{f(x_1) - f(x_2)}^\rho d\gamma(x_1, x_2) \nonumber \\
        & \qquad \text{(By the fact that $\Gamma(\probQ, \Delta_{\bsR, \bsC}\#\prob)\subseteq S_{\theta+\theta_d}(\Delta_{\bsR, \bsC} \# \prob)$)} \nonumber\\
        &\leq \sup_{\gamma \in S_{\theta+\theta_d}(\Delta_{\bsR, \bsC} \# \prob)} \int_{\sX \times \sX} \norm{f(x_1) - f(x_2)}^\rho d\gamma(x_1, x_2)\label{ineq:s-space}
    \end{align}
    Applying Lagrangian duality to \eqref{ineq:s-space}:
    \begin{align}
        &\sup_{\gamma \in S_{\theta+\theta_d}(\Delta_{\bsR, \bsC} \# \prob)} \int_{\sX \times \sX} \norm{f(x_1) - f(x_2)}^\rho d\gamma(x_1, x_2) \nonumber \\ &\qquad \text{(By Lagrangian strong duality)} \nonumber \\&= \inf_{\lambda \geq 0} \sup_{\gamma \in \sP(\sX \times \sX)} \bigg\{ \int_{\sX \times \sX} \norm{f(x_1) - f(x_2)}^\rho d\gamma(x_1, x_2) + \nonumber \\&\lambda \bigg( (\theta+\theta_d)^\rho - \int_{\sX \times \sX} \norm{x_1 - x_2}^\rho d\gamma(x_1, x_2) \bigg) : \maybenonumber\maybenewline\maybeamp \text{proj}_2 \# \gamma = \Delta_{\bsR, \bsC} \# \prob \bigg\} \label{eq:46} \\
        &\qquad \text{(Reorganizing the terms)} \nonumber \\
        &= \inf_{\lambda \geq 0} \sup_{\gamma \in \sP(\sX \times \sX)} \bigg\{ \lambda (\theta+\theta_d)^\rho + \nonumber \\
        & \int_{\sX \times \sX} \big( \norm{f(x_1) - f(x_2)}^\rho - \lambda \norm{x_1 - x_2}^\rho \big)d\gamma(x_1, x_2) \maybenonumber\maybenewline
        \maybeamp: \text{proj}_2 \# \gamma = \Delta_{\bsR, \bsC} \# \prob \bigg\} \nonumber \\
        &\qquad \text{(By applying Theorem 1 from \cite{gao2023distributionally})} \nonumber \\
        &= \inf_{\lambda \geq 0} \bigg\{ \lambda (\theta+\theta_d)^\rho + \maybenonumber\maybenewline
        \maybeamp \int_{\sX} \sup_{\xi \in \sX} \big( \norm{f(\xi) - f(\varsigma)}^\rho - \lambda \norm{\xi - \varsigma}^\rho \big)d\big(\Delta_{\bsR, \bsC} \# \prob \big)(\varsigma) \bigg\} \label{eq:48} \\
        &\qquad \text{(By using the definition of $\Delta_{\bsR, \bsC}\#\prob$)} \nonumber \\
        &= \inf_{\lambda \geq 0} \bigg\{ \lambda (\theta+\theta_d)^\rho \maybenonumber\maybenewline 
        \maybeamp + \sum_{j=1}^N \evpi{j} \sup_{\xi \in \sX} \big( \norm{f(\xi) - f(\vc_j)}^\rho - \lambda \norm{\xi - \vc_j}^\rho \big) \bigg\} \nonumber
    \end{align}
    \ifthenelse{\boolean{doublecolumn}}{\qed}{}
\end{proof}

We are now ready to prove Theorem \ref{th:bound}. By Proposition \ref{proposition:lagrangian-duality},
\begin{align}
    &\left( \sup_{\probQ \in \mathbb{B}_\theta(\prob)} \wasserstein_\rho(f\#\probQ, f\#\Delta_{\bsR, \bsC}\#\prob) \right)^\rho \nonumber\\
    &\leq \inf_{\lambda \geq 0} \bigg\{ \lambda (\theta+\theta_d)^\rho \maybenonumber\maybenewline
    \maybeamp\maybeqquad + \sum_{j=1}^N \evpi{j} \sup_{\xi \in \sX} \big( \norm{f(\xi) - f(\vc_j)}^\rho - \lambda \norm{\xi - \vc_j}^\rho \big) \bigg\} \nonumber \\
    &\qquad \text{(By the norm linearization in \eqref{eq:norm-linearization})} \nonumber \\
    &\leq \inf_{\lambda \geq 0} \bigg\{ \lambda (\theta+\theta_d)^\rho \maybenonumber\maybenewline
    \maybeamp\maybeqquad + \sum_{j=1}^N \evpi{j} \sup_{\xi \in \sX} \big( \alpha_j \norm{\xi - \vc_j}^\rho + \beta_j - \lambda \norm{\xi - \vc_j}^\rho \big) \bigg\} \nonumber \\
    &= \inf_{\lambda \geq 0} \bigg\{ \lambda (\theta+\theta_d)^\rho \maybenonumber\maybenewline
    \maybeamp\maybeqquad + \sum_{j=1}^{N} \evpi{j}\beta_j + \sum_{j=1}^N \evpi{j} \sup_{\xi \in \sX} \big( (\alpha_j-\lambda) \norm{\xi - \vc_j}^\rho \big) \bigg\} \label{eq:inf-for-lambda-proof}
\end{align}
First, consider the case where $\sX$ is unbounded (e.g. $\sX = \realNum^d$). If there exists $\alpha_\ell$ such that $\alpha_\ell > \lambda$, then the correspondent inner supremum returns $\infty$. Thus, in the outer minimization program, it is enough to search for $\lambda \geq \max_{j \in \{1,...,N\}} \alpha_j$. Moreover, we note that for any $\lambda \geq \max_{j \in \{1,...,N\}} \alpha_j$, the inner supremum returns $0$. Hence, the solution of the whole optimization program is given by $\lambda^{*} = \max_{j \in \{1,...,N\}} \alpha_j$, so that \eqref{eq:inf-for-lambda-proof} becomes
  $  \lambda^{*} (\theta+\theta_q)^\rho + \sum_{j=1}^{N} \evpi{j}\beta_j. $
For bounded $\sX$, one may find a solution $\lambda^* < \max_{j \in \{1,...,N\}} \alpha_j$ for the entire optimization program. However, we remark that choosing $\Tilde{\lambda} = \max_{j \in \{1,...,N\}} \alpha_j$ still provides a valid upper bound on \eqref{eq:inf-for-lambda-proof}. 
\qed

\subsection{Proof of Theorem \ref{th:bound-zero-ball}}\label{proof:th:bound-zero-ball}
Let $\gamma^*$ be defined as in \eqref{eq:general-coupling}. Then, from statement i) from Lemma \ref{lemma:coupling-signature}, we have:
\begin{align}
    \maybeamp\wasserstein_\rho(f\#\prob, f\#\Delta_{\bsR, \bsC}\#\prob)^\rho \maybenonumber\maybenewline  
    &\quad \leq \int_{\sX \times \sX} \norm{f(x)-f(y)}^{\rho} d\gamma^*(x, y) \label{eq:36} \\ 
    &\quad = \int_{\sX \times \sX} \norm{f(x)-f(y)}^{\rho} \sum_{k=1}^N\indicator_{\sR_k}(\vx) d\prob(x) d\delta_{c_k}(y) \nonumber \\ 
    &\quad =\sum_{k=1}^N \int_{\sR_k} \norm{f(x)-f(c_k)}^\rho d\prob(x) \nonumber \\
    & \qquad \text{(By the norm linearization in \eqref{eq:norm-inequality-zero-budget}))}\nonumber\\
    &\quad \leq \sum_{k=1}^N \int_{\sR_k} \big( \alpha_k \norm{x-c_k}^\rho + \beta_k \big) d\prob(x) \nonumber \\
    &\quad =\sum_{k=1}^N \alpha_k \int_{\sR_k} \norm{x-c_k}^\rho d\prob(x) + \sum_{k=1}^N \beta_k \prob(\sR_k) \nonumber  
\end{align}
In the case where $\bsR$ is the Voronoi partition w.r.t. $\bsC$, by item ii) of Lemma \ref{lemma:coupling-signature}, the inequality in \eqref{eq:36} can be replaced by equality. The rest of the proof remains the same.
\qed

\subsection{Proof of Lemma \ref{prop:bound-on-convergence-problem-1}}\label{proof:prop:bound-on-convergence-problem-1}
By straightforward applications of the triangle inequality:
\begin{align*}
    \maybeamp\sup_{\probQ \in \mathbb{B}_\theta(\prob)} \wasserstein_{\rho}(f \# \probQ, f \# \Delta_{\bsR, \bsC} \# \prob) \maybenewline 
    \maybeamp\leq\sup_{\probQ \in \mathbb{B}_\theta(\prob)} \wasserstein_{\rho}(f \# \probQ, f  \# \prob) + \wasserstein_{\rho}(f  \# \prob, f \# \Delta_{\bsR, \bsC} \# \prob) ,
\end{align*}
\begin{align*}
    \maybeamp\sup_{\probQ \in \mathbb{B}_\theta(\prob)} \wasserstein_{\rho}(f \# \probQ, f \# \prob) \maybenewline\maybeamp\leq  \sup_{\probQ \in \mathbb{B}_\theta(\prob)} \wasserstein_{\rho}(f \# \probQ, f  \# \Delta_{\bsR, \bsC} \# \prob) + \wasserstein_{\rho}(f \# \Delta_{\bsR, \bsC}  \# \prob, f \#\prob). 
\end{align*}
We conclude by combining both inequalities.
\qed

\subsection{Proof of Theorem \ref{prop:convergence-algorithm-error-approx}}\label{proof:prop:convergence-algorithm-error-approx}
From Lemma \ref{prop:bound-on-convergence-problem-1}, to prove this theorem, it is enough to show that
$    \wasserstein_{\rho}(f  \# \prob, f \# \Delta_{\bsR^*, \bsC^*} \# \prob) \leq \epsilon. $
From Theorem \ref{th:bound-zero-ball}, by taking $(\alpha_k, \beta_k) = (\sL_f, 0)$ as discussed in Remark \ref{rem:linearSyst}, we have that:
\begin{align}
     \maybeamp\wasserstein_{\rho}(f  \# \prob, f \# \Delta_{\bsR^*, \bsC^*} \# \prob)^\rho 
    &\leq \sL_f^\rho \sum_{k=1}^{N+1} \int_{\sR_k^*} \norm{x-c_k^*}^\rho d\prob(x) \nonumber \\
    &= \sL_f^\rho \sum_{k=1}^N \int_{\sR_k} \norm{x-c_k}^\rho d\prob(x) + \sL_f^\rho \int_{\sX \setminus \bar\sX} \norm{x-\bar{c}}^\rho d\prob(x) \nonumber \\
    &\leq \sL_f^\rho \sum_{k=1}^N \int_{\sR_k} \norm{x-c_k}^\rho d\prob(x) + \frac{\epsilon^\rho}{2} \label{eq:63}
\end{align}
where we use the fact that, by construction, $\bsR^* \coloneqq \bsR \cup \{ \sX \setminus \bar\sX \}$ and $\bsC^* \coloneqq \bsC \cup \{ \bar{c} \}$, and also $\int_{\sX \setminus \bar\sX} \norm{x-\bar{c}}^\rho d\prob(x) \leq \frac{\epsilon^\rho}{2\sL_f^\rho}$ (which, we must highlight, is always possible as $\prob \in \sP_\rho(\sX)$). Then, what is left to show is that the left term in \eqref{eq:63} can also be upper-bounded by $\frac{\epsilon^\rho}{2}$. Indeed, because $\norm{R_k}_\infty = \frac{\norm{\bar\sX}_\infty}{N^\frac{1}{d}}$ (as all compact regions are hypercubic), it holds that:
$
    \norm{R_k}_\infty = \frac{\norm{\bar\sX}_\infty}{N^\frac{1}{d}} \leq \frac{\epsilon}{2^\frac{1}{\rho} d^\frac{1}{\rho} \sL_f},
$
where we use the fact that again by construction, $N \geq \left( \frac{2^\frac{1}{\rho} \sL_f d^\frac{1}{\rho} \norm{\bar\sX}_\infty}{\epsilon} \right)^d$. 
Thus, 
\begin{align*}
    &\sL_f^\rho \sum_{k=1}^N \int_{\sR_k} \norm{x-c_k}^\rho d\prob(x) \nonumber \\
    &\qquad \text{(From the $L_\rho$-norm definition)} \\
    &\qquad = \sL_f^\rho \sum_{k=1}^N \int_{\sR_k} \sum_{i=1}^d |x^{(i)}-c_k^{(i)}|^\rho d\prob(x) \\
    &\qquad \text{(From the $\norm{.}_\infty$ definition)} \\
    &\qquad \leq \sL_f^\rho \sum_{k=1}^N \int_{\sR_k} \sum_{i=1}^d \norm{\sR_k}_\infty^\rho d\prob(x) \\
    &\qquad \text{(Using that $\norm{R_k}_\infty \leq \frac{\epsilon}{2^\frac{1}{\rho} d^\frac{1}{\rho} \sL_f}$)} \\
    &\qquad \leq \sL_f^\rho \sum_{k=1}^N \int_{\sR_k} \sum_{i=1}^d \frac{\epsilon^\rho}{2d \sL_f^\rho} d\prob(x) \maybenewline
    \maybeamp\maybeqquad = \sL_f^\rho \sum_{k=1}^N \frac{\epsilon^\rho}{2 \sL_f^\rho} \prob(\sR_k) 
     = \frac{\epsilon^\rho}{2}  \sum_{k=1}^N \prob(\sR_k)
    \leq \frac{\epsilon^\rho}{2}.
\end{align*}
\qed

\subsection{Proof of Theorem \ref{thm:how-to-propagate-stochastic-systems}}\label{proof:thm:how-to-propagate-stochastic-systems}
We use the same notation as in Figure \ref{fig:approximation-scheme-dynamical-systems}. The proof follows by induction. The base case is $t=1$, for which we have
\begin{align*}
    \wasserstein_\rho(\prob_{x_1}, \probHat_{x_1}) \maybeamp= \wasserstein_\rho(f\#\prob_0, f\#\Delta_{\bsR_0, \bsC_0}\#\probHat_0) \maybenewline\maybeamp= \wasserstein_\rho(f\#\probHat_0, f\#\Delta_{\bsR_0, \bsC_0}\#\probHat_0)
\end{align*}
since $\prob_0 = \probHat_0$ (as $\prob_{x_0} = \probHat_{x_0}$). Thus, the bound $\theta_1$ comes from the application of Theorem \ref{th:bound-zero-ball}.
For the induction case (i.e., $t>1$) we have:
\begin{align*}
    \wasserstein_\rho(\prob_{x_{t+1}}, \probHat_{x_{t+1}}) \maybeamp= \wasserstein_\rho(f\#\prob_t, f\#\Delta_{\bsR_t, \bsC_t}\#\probHat_t) \maybenonumber\maybenewline\maybeamp\leq \sup_{\probQ \in \mathbb{B}_{\theta_t}(\probHat_{x_t})}\wasserstein_\rho(f\#\probQ, f\#\Delta_{\bsR_t, \bsC_t}\#\probHat_t),
\end{align*}
from which the bound $\theta_{t+1}$ follows from applying Theorem \ref{th:bound} for $\theta = \theta_t$, and using that $\theta_{d,t} \leq \epsilon$. This proves statement i) in the Theorem. For statement ii), we first note that by Remark \ref{rem:linearSyst}, for $t>1$:
\begin{align*}
    \wasserstein_\rho(\prob_{x_{t+1}}, \probHat_{x_{t+1}}) \maybeamp\leq \theta_{t+1} \maybenewline\maybeamp\leq \left(  \alpha_{\text{max},t}( \theta_t + \epsilon )^\rho + \sum_{k=1}^{N_t} \prob(\sR_{t,k}) \beta_{t,k} \right)^\frac{1}{\rho} \maybenewline\maybeamp\leq \sL_f(\theta_t + \epsilon)
\end{align*}
Let $T:\realNum \to \realNum$ be a map given by $T(\theta) \coloneqq \sL_f(\theta + \epsilon)$. Note that $T$ is contractive since $|T(\theta)-T(\Tilde{\theta})| \leq |\sL_f(\theta - \Tilde{\theta})| \leq \sL_f|\theta - \Tilde{\theta}|$. One can easily find a fixed point $\theta^*$ for $T$, i.e.
\begin{equation*}
    \theta^* = T(\theta^*) \iff \theta^* = \sL_f(\theta^* + \epsilon) \iff \theta^* = \frac{\sL_f}{1 - \sL_f} \epsilon
\end{equation*}
Then, by the Banach fixed-point theorem, for the sequence $\theta_{t+1} = T(\theta_t)$, it holds that $\lim_{t\to\infty}\theta_t = \theta^*$, which concludes the proof.
\qed

\section{Implementation Details}
\begin{table*}[t]
% \begin{scriptsize}
 \resizebox{\textwidth}{!}{%
\begin{tabular}{lcccccccc}
    \textbf{} & \textbf{} & \textbf{} & \multicolumn{6}{c}{\textbf{Section}} \\ \cline{4-9} 
    \textbf{} & \textbf{} & \textbf{} & \multicolumn{1}{c|}{\textbf{7.1}} & \multicolumn{5}{c}{\textbf{7.4}} \\ \hline
    \textbf{System} & $d$ & $f(x)$ & \multicolumn{1}{c|}{$\prob$}   & $f(x, \omega)$ & $\prob_{x_0}$ & $\prob_\omega$ & $|\bsC|$ & $T$ \\ \hline \\
    Sigmoid & $1$ & $f_\text{Sigm}$ & \multicolumn{1}{c|}{$\sN(0.2, 0.5)$} & & & & & \\ \\
    Bounded Linear & $2$ & $f_\text{BoundLin}$ & \multicolumn{1}{c|}{$\sN\bigg(\begin{bmatrix}1.5 \\2.5\end{bmatrix}, \begin{bmatrix}0.4 & 0.0 \\0.0 & 0.5\end{bmatrix} \bigg)$} & & & & & \\ \\
    Quadruple-Tank & $4$ & $f_\text{QuadTank}$ & \multicolumn{1}{c|}{$\sN\left(\begin{bmatrix}1.5 \\2.5\\-0.5\\-1.0\end{bmatrix}, \begin{bmatrix}0.001 & 0 & 0 & 0 \\0 & 0.02 & 0 & 0 \\ 0 & 0 & 0.4 & 0 \\ 0 & 0 & 0 & 0.01 \end{bmatrix} \right)$} & & & & & \\ \\
    NN Layer & $10$ & $f_\text{NNLay}$ & \multicolumn{1}{c|}{$\bar{\prob}$} & $\sigma(Ax+B\omega)$ & $\bar{\prob}_{x_0}$ & $\bar{\prob}_{\omega}$ & $10^2$ & $50$ \\ \\
    Mountain Car & $2$ & $f_{\text{MountCar}}$ & \multicolumn{1}{c|}{$\sN\left(\begin{bmatrix}0.3 \\0.2\end{bmatrix}, \begin{bmatrix}10^{-1} & 0 \\0 & 10^{-3}\end{bmatrix}\right)$} & $f(x) + \omega$ & $\prob$ (from 7.1) & $\sN\left(\begin{bmatrix}0 \\0\end{bmatrix}, 10^{-2}I\right)$ & $10^2$ & $50$ \\ \\
    Dubins Car & $3$ & $f_\text{DubinsCar}$ & \multicolumn{1}{c|}{$\sN\left(\begin{bmatrix}0.3 \\0.2\\0.01\end{bmatrix}, \begin{bmatrix}10^{-1} & 0 & 0 \\0 & 10^{-2} & 0\\ 0 & 0 & 10^{-3} \end{bmatrix}\right)$} & $f(x) + \omega$ & $\prob$ (from 7.1) & $\sN\left(\begin{bmatrix}0 \\0\\0\end{bmatrix}, 10^{-2}I\right)$ & $10^2$ & $50$ \\ \hline
\end{tabular}
% \end{scriptsize}
 }
\caption{Summary of implementation details}
\label{table:details-summary}
\end{table*}

In the following, we present the implementation details of the experiments in Section \ref{section:experimental-results}. First, we introduce the piecewise Lipschitz continuous functions $f$ that we consider. Then, in Table \ref{table:details-summary}, we show the probability distributions used in the experiments.

\subsubsection{Functions}\label{appendix:functions}
\paragraph{Sigmoid (Example \ref{exampl:ImportanceOdAlphaBeta})}$f_{\text{Sigm}} = \frac{1}{1+e^{-x}}.$

\paragraph{Bounded Linear (adapted from \cite{santoyo2021barrier})} $f_{\text{BoundLin}}:\realNum^2 \rightarrow \realNum^2$ such that
\begin{equation}
    f_{\text{BoundLin}}(x) = \text{clamp}\bigg( \begin{bmatrix}0.0 & 0.4 \\0.3 & 0.8\end{bmatrix}x, \begin{bmatrix}-2 \\-2\end{bmatrix}, \begin{bmatrix}2 \\2\end{bmatrix} \bigg).
\end{equation}

\paragraph{Quadruple-Tank (instance of \cite{johansson2000quadruple})} $f_{\text{QuadTank}}:\realNum^4 \rightarrow \realNum^4$ such that:
\begin{equation}
    f_{\text{QuadTank}}(x) = \begin{bmatrix}0.721 & 0 & 0.041 & 0 \\ 0 & 0.718 & 0 & 0.033 \\ 0 & 0 & 0.724 & 0 \\ 0 & 0 & 0 & 0.737 \end{bmatrix}x
\end{equation}

\paragraph{NN Layer}$f_{\text{NNLay}}:\realNum^{10} \rightarrow \realNum^{10}$ such that $
    f_{\text{NNLay}}(x) = \sigma(Ax)$,
where $A = \text{diag}(3 \times 10^0, 10^{-3}, 5 \times 10^{-3}, 7 \times 10^{-3}, 3 \times 10^{-2}, 10^{-3}, 10^{-3}, 10^{-3}, 10^{-3}, 10^{-3})$.

\paragraph{Mountain Car (adapted from \cite{singh1996reinforcement})} $f_{\text{MountCar}}:\realNum^{2} \rightarrow \realNum^{2}$ such that:
\begin{align}
    f_{\text{MountCar}}(x) &= \text{clamp}\left( \begin{bmatrix}1 & 0 \\ 1 & 1\end{bmatrix}x + \begin{bmatrix}10^{-3} \\ 0\end{bmatrix}, \begin{bmatrix}-0.5 \\ -0.5\end{bmatrix}, \begin{bmatrix}1.2 \\ 1.2\end{bmatrix} \right) \maybenonumber\maybenewline\maybeamp \maybeqquad - 2.5 \times 10^{-3}\begin{bmatrix}\cos{\big(3x^{(2)}\big)} \\ 0\end{bmatrix}
\end{align}

\paragraph{Dubins Car \cite{balkcom2018dubins}} $f_{\text{DubinsCar}}:\realNum^{3} \rightarrow \realNum^{3}$ such that
\begin{equation}
        f_{\text{DubinsCar}}(x) = \begin{bmatrix}x^{(1)} + 1.5\sin{\big(x^{(3)}\big)} \\ 
    x^{(2)} + 1.5\cos{\big( x^{(3)} \big)} \\ x^{(3)}+ 0.6 \end{bmatrix}.
\end{equation}

\subsection{Further details}\label{subsection:further-details}
In Section \ref{subsection:convergence}, for the NN Layer function, we consider (see Table \ref{table:details-summary}) $\bar{\prob}_{x_0} = \sN(\mu_{NN}, \Sigma_{NN})$, with:
$$\mu_{NN} = [0.0, 1.0, 0.5, -0.7, 0.3, 2.0, -3.0, 0.4, -0.1, 4.0]^T,$$ and $$
\Sigma_{NN} = \text{diag}([0.0001, 0.5, 0.7, 0.2, 1.5, 2.5, 0.1, 0.5, 0.8, 0.2]).$$
Alternatively, in Section \ref{subsection:stochastic-system-experiments-results}, we consider a 3D NN Layer, where (using the same notation as in Table \ref{table:details-summary})
$A = \text{diag}([3.0, 1.5, 1.2]), B = \text{diag}([0.5, 1.0, 0.9])$,
\begin{align*}
    \maybeamp\bar{\prob}_{x_0} = \sN\left(\begin{bmatrix}1.5 \\-1.2\\2.4\end{bmatrix}, \begin{bmatrix}10^{-1} & 0 & 0 \\0 & 0.5 & 0\\ 0 & 0 & 0.2 \end{bmatrix}\right), \qquad \maybenewline
    \maybeamp\bar{\prob}_\omega = \sN\left(\begin{bmatrix}0 \\0\\0\end{bmatrix}, \begin{bmatrix}10^{-1} & 0 & 0 \\0 & 10^{-1} & 0\\ 0 & 0 & 10^{-2} \end{bmatrix}\right).
\end{align*}
Finally, in Section \ref{subsection:applying-loc-selection-algo}, we consider the following functions and distributions, for $d \in \{ 1, 2, 3, 4\}$:
For $d=1$, $f(x) = \text{clamp}(3x, -1, 1),\prob = \sN\left(0, 1\right).$
For $d=2$, $f(x) = \text{clamp}(\text{diag}([3, 0.001])x, -2, 2),$ 
$$
\prob = \sN\left(\begin{bmatrix}3 \\1\end{bmatrix}, \begin{bmatrix}0.02 & 0 \\0 & 0.5 \end{bmatrix}\right).$$
For $d=3$, $f(x) = \text{clamp}(\text{diag}([3, 0.001, 1.1])x, -2, 2),$ $$\prob = \sN\left(\begin{bmatrix}3 \\1\\-0.9\end{bmatrix}, \begin{bmatrix}0.02 & 0 &0 \\0 & 0.5&0 \\0&0&0.001 \end{bmatrix}\right).$$
For $d=4$, $f(x) = \text{clamp}(\text{diag}([3, 0.001, 1.1, 2.2])x, -2, 2),$ $$\prob = \sN\left(\begin{bmatrix}3 \\1\\-0.9\\0.4\end{bmatrix}, \begin{bmatrix}0.02 & 0 &0&0 \\0 & 0.5&0&0 \\0&0&0.001&0 \\ 0&0&0&0.2 \end{bmatrix}\right).$$

%==============================================================================
\clearpage
\bibliographystyle{plain}%{apalike} %{siam}
\begin{small}
\bibliography{main}
\end{small}

\end{document}